%% file: ola_paper.tex
\documentclass[11pt]{article}
\usepackage{hyperref}
\usepackage{fullpage}
\usepackage[english]{babel}
\usepackage[cmex10]{amsmath}
\usepackage{microtype}
\usepackage{xspace}
\usepackage{comment}
\usepackage{letltxmacro}
\usepackage{comment}
\usepackage{enumerate}
\usepackage{enumitem}
\usepackage{tikz}
\usetikzlibrary{calc,arrows,automata,decorations.pathreplacing}
\tikzstyle{vertex}=[circle,draw=black,fill=black,minimum size=3.8pt,inner sep=0pt]
\tikzstyle{dot}=[circle,fill=black,minimum size=0.5pt,inner sep=0pt]
\definecolor{greenedgecolor}{rgb}{0.0,0.45,0.0}
\tikzstyle{greenedge} = [draw,thick,greenedgecolor]
\definecolor{blueedgecolor}{rgb}{0,0.15,0.85}
\tikzstyle{blueedge} = [draw,thick,blueedgecolor] 
\definecolor{graybluecolor}{rgb}{0.545, 0.6353, 0.7255}
\definecolor{yellowishcolor}{rgb}{1, 0.9, 0.3}
\definecolor{deepbluecolor}{rgb}{0, 0.2, 0.4}
\definecolor{darkgraybluecolor}{rgb}{0, 0.2, 0.4}
\usepackage{url}
\usepackage{graphicx}
\usepackage{caption}
\usepackage{subfigure}
\usepackage{float}

\usepackage{ifthen}
\usepackage{epsfig}
\usepackage{amsfonts}
\usepackage{amssymb}
\usepackage{amstext}
\usepackage{amsmath}
\usepackage{amsthm}
\usepackage{comment}
\usepackage{tikz}
\usepackage{algorithm}
\usepackage{float}
\usepackage{caption}
\usepackage{cite}
\usepackage{color}
\usepackage{graphics}
\usepackage{url}
\usepackage{framed}
\usepackage{listings} 
\usepackage{xcolor}
\usepackage{xspace}
\usepackage[noend]{algpseudocode}
\usepackage[titletoc]{appendix}
\usepackage[none]{hyphenat}

\newfloat{algorithm}{tp}{toa}
\usetikzlibrary{positioning,chains,fit,shapes,calc}

\newtheorem{theorem}{Theorem}[section]
\newtheorem{lemma}[theorem]{Lemma}
\newtheorem{claim}[theorem]{Claim}

\newtheorem{hypothesis}[theorem]{Hypothesis}

\newtheorem{definition}[theorem]{Definition}

\makeatletter
\renewcommand{\@cite}[1]{[#1]}
\makeatother 

\makeatletter
\newtheorem*{rep@theorem}{\rep@title}
\newcommand{\newreptheorem}[2]{%
	\newenvironment{rep#1}[1]{%
		\def\rep@title{#2 \ref{##1}}%
		\begin{rep@theorem}}%
		{\end{rep@theorem}}}
\makeatother
\newreptheorem{theorem}{Theorem}
\newreptheorem{hypothesis}{Hypothesis}

\newcommand{\problem}[3]{
\begin{framed} 
  \noindent
  #1 \\
  {\bf Input:}#2\\
  {\bf Question:}#3
\end{framed} 
}

\newcommand{\gapproblem}[4]{
\begin{framed} 
  \noindent
  #1 \\
  {\bf Input:}#2\\
  {\bf Case 1:} #3\\
  {\bf Case 2:} #4
\end{framed} 
}

\tikzstyle{vertG}=[circle,draw=black,fill=white,minimum size=15pt,inner sep=0pt]
\tikzstyle{vertH}=[circle,draw=black,fill=black,minimum size=11pt,inner sep=0pt]
\tikzstyle{dot}=[circle,draw=black,fill=black,minimum size=2pt,inner sep=0pt]
\tikzstyle{oedge} = [draw,thick,->]
\tikzstyle{edge} = [draw,thick]
\tikzstyle{selected edge} = [draw,line width=5pt,-,red!50]
\tikzstyle{weight} = [font=\small]	

\ifthenelse{ \equal{0}{1} }{
\newcommand\todo[1]{(\textit{\underline{TODO:} #1})}
\newcommand\expl[1]{(\textit{\underline{EXPL:} #1})}
\newcommand\bigtodo[1]{
\begin{framed}
(\textit{\underline{TODO:} #1})
\end{framed}
}
}
{
\newcommand\todo[1]{}
\newcommand\expl[1]{}
\newcommand\bigtodo[1]{}
}

\def\N{\mathbb{N}}

\def\to{\rightarrow}

\def\ss{\subseteq}

\def\vertices[#1]{V({#1})}
\def\edges[#1]{E({#1})}

\def\minb{{\sc Min Bisection}\xspace} 
\def\ola{{\sc Optimum Linear Arrangement}\xspace} 
\def\olad[#1]{\text{\sc Optimum Linear Arrangement$_{\leq}$($#1$)}\xspace} 
\def\fas{\text{{\sc Feedback} {\sc Arc} {\sc Set}}\xspace}
\def\fvs{\text{\sc Feedback Vertex Set}\xspace}
\def\fast{\text{\sc Feedback} {\sc Arc} {\sc Set} {\sc in} {\sc Tournaments}\xspace}
\def\olashort{\text{\sc OLA}\xspace} 
\def\oladshort[#1]{\text{\sc OLA$_{\leq}$($#1$)}\xspace} 
\def\fasshort{\text{\sc FAS}\xspace}
\def\fvsshort{\text{\sc FVS}\xspace}
\def\fvss[#1]{\text{\sc FVS-BAL($#1$)}\xspace}
\def\fastshort{\text{\sc FAST}\xspace}
\def\justsat{\text{\sc Sat}\xspace}
\def\xsat{\text{\sc X-Sat}\xspace}
\def\sat[#1]{\text{\sc #1-Sat}\xspace}
\def\ssat[#1]#2{\text{\sc (E$3_{#1}$,E$2_{#2}$)-SAT}\xspace}
\def\esat[#1]{\text{\sc E#1-Sat}\xspace}
\def\enaesat[#1]{\text{\sc E#1-NAE-Sat}\xspace}
\def\andsat[#1]{\text{\sc #1-And-Sat}\xspace}
\def\maxcut{\text{\sc Max Cut}\xspace}
\def\chaincompletion{\text{\sc Chain} {\sc Completion}\xspace}
\def\minimumfillin{\text{\sc Minimum} {\sc Fill-In}\xspace}
\def\intervalcompletion{\text{\sc Interval} {\sc Completion}\xspace}
\def\properintervalcompletion{\text{\sc Proper} {\sc Interval} {\sc Completion}\xspace}
\def\triviallyperfectcompletion{\text{\sc Trivially} {\sc Perfect} {\sc Completion}\xspace}
\def\thresholdcompletion{\text{\sc Threshold} {\sc Completion}\xspace}

\def\gapminb{{\sc Gap Min Bisection}} 
\def\gapminbd[#1]{{\sc Gap Min Bisection($#1$)}} 
\def\gapmaxcut{\text{\sc Gap Max Cut}} 
\def\gapmmaxcut{\text{\sc Gap Multigraph Max Cut}} 
\def\gapesat[#1]{\text{\sc Gap E#1-SAT}} 
\def\justgapfas{\text{\sc Gap FAS}} 
\def\gapfas[#1]{\text{\sc Gap FAS($#1$)}} 
\def\gapfvs[#1]{\text{\sc Gap FVS-BAL($#1$)}} 
\def\gapssat[#1]#2{\text{\sc Gap (E$3_{#1}$,E$2_{#2}$)-SAT}}
\def\gapenaesat[#1]{\text{\sc Gap E#1-NAE-SAT}}
\def\gapand[#1]{\text{\sc Gap #1-AND-SAT}}

\newcommand\half[1]{\frac{#1}{2}}

\def\bounda{\alpha}
\def\boundb{\beta}
\def\dH{d_H}
\def\dG{d_G}

\def\dHG{\Delta_{H, G}}
\def\dHi{d_{H_i}}

\def\dHip{d_{H_{i+1}}}
\def\phin{\lceil \varphi n \rceil}
\def\order{\sigma}
\def\aorder{\pi}

\newcommand\term[1]{{\em #1}}

\newcommand{\Oh}{{\mathcal{O}}}

\newcommand{\linleq}{\leq_{\textrm{P}}^{\textrm{lin}}}

\title{Lower bounds for the parameterized complexity of Minimum Fill-In and other completion problems}
\author{%
Ivan Bliznets\thanks{St.~Petersburg Department of Steklov Institute of Mathematics. E-mail: \texttt{iabliznets@gmail.com}. Partially supported by the Government of the Russian Federation (grant 14.Z50.31.0030), by the Grant of the President of the Russian Federation (MK-6550.2015.1) as well as by Warsaw Center of Mathematics and Computer Science.} 
\and
Marek Cygan\thanks{Institute of Informatics, University of Warsaw, Poland. E-mail: \texttt{cygan@mimuw.edu.pl}. Partially supported by the Polish National Science Centre grant DEC-2012/05/D/ST6/03214.} 
\and
Pawe\l{} Komosa\thanks{Institute of Informatics, University of Warsaw, Poland. E-mail: \texttt{kompaw01@gmail.com}.}
\and
Luk\'a\v s Mach\thanks{DIMAP and Department of Computer Science, University of Warwick, United Kingdom. E-mail: \texttt{lukas.mach@gmail.com}. Received funding by the European Research Council under the European Union's Seventh Framework Programme (FP7/2007-2013)/ERC grant agreement no.~259385.}
\and
Micha\l{} Pilipczuk\thanks{Institute of Informatics, University of Warsaw, Poland. E-mail: \texttt{mp248287@mimuw.edu.pl}. Supported by the Polish National Science Centre grant DEC-2013/11/D/ST6/03073 and Foundation for Polish Science via the START stipend programme. During the work on these results, Micha\l{} Pilipczuk has been holding a post-doc position of Warsaw Center of Mathematics and Computer Science.}}
\date{}

\begin{document}

\pagenumbering{gobble}
\thispagestyle{empty}

\date{}

\maketitle
\begin{abstract}
\input{sec_abstract}

\end{abstract}
\clearpage

\pagenumbering{arabic}

\input{sec_introduction} 
\input{sec_prelim}

\input{sec_nosubexp} 
\input{sec_sparse}

\input{sec_sparse_theorem}

\input{sec_hardness}
\input{sec_fast}

\input{sec_conclusions}

\bibliographystyle{abbrv}
\bibliography{ola_paper}

\begin{appendices}
	\input{sec_problems}

\end{appendices}

\end{document}

%% file: sec_abstract.tex
In this work, we focus on several completion problems for subclasses of chordal graphs: \minimumfillin, \intervalcompletion, \properintervalcompletion, \thresholdcompletion,
and \triviallyperfectcompletion. In these problems, the task is to add at most $k$ edges to a given graph in order to obtain a chordal, interval, proper interval, threshold, or trivially perfect graph, respectively. We prove the following lower bounds for all these problems, as well as for the related \chaincompletion problem:
\begin{itemize}
\item Assuming the Exponential Time Hypothesis, none of these problems can be solved in time $2^{\Oh(n^{1/2}/\log^c n)}$ or $2^{\Oh(k^{1/4}/\log^c k)}\cdot n^{\Oh(1)}$, for some integer $c$.
\item Assuming the non-existence of a subexponential-time approximation scheme for \minb on $d$-regular graphs, for some constant $d$, none of these problems can be solved in time $2^{o(n)}$ or $2^{o(\sqrt{k})}\cdot n^{\Oh(1)}$.
\end{itemize}
For all the aforementioned completion problems, apart from \properintervalcompletion, FPT algorithms with running time of the form $2^{\Oh(\sqrt{k}\log k)}\cdot n^{\Oh(1)}$ are known. Thus, the second result proves that a significant improvement of any of these algorithms would lead to a surprising breakthrough in the design of approximation algorithms for \minb.

To prove our results, we use a reduction methodology based on combining the classic approach of starting with a sparse instance of \sat[3], prepared using the Sparsification Lemma, with the existence of almost linear-size Probabilistically Checkable Proofs (PCPs). Apart from our main results, we also obtain lower bounds excluding the existence of subexponential algorithms for the \ola problem, as well as improved, yet still not tight, lower bounds for \fast.

%% file: sec_introduction.tex
\section{Introduction}

In the \minimumfillin problem, also known as {\sc{Chordal Completion}}, the input is an undirected graph $G$ and an integer $k$, and the question is whether at most $k$ edges can be added to $G$ in order to turn it into a chordal graph, i.e., a graph without induced cycles of length at least $4$ (also known as {\em{holes}}). The interest in this problem originates in the study of strategies for Gaussian elimination on sparse matrices, because the optimum number of additional entries of a matrix $A$ that become non-zero during the elimination is tightly connected to the minimum fill-in of the graph $G_A$ obtained by taking $A$ to be its adjacency matrix. See~\cite{davis-book,Rose} for more information on applications of \minimumfillin in the theory of sparse matrices. However, the problems of adding as few edges as possible to obtain a chordal graph, or a graph belonging to some natural subclass of chordal graphs, like interval, proper interval, trivially perfect, or threshold graphs, have numerous other applications ranging from database management, bioinformatics, artificial intelligence, to social networks. We refer to the introductory sections of~\cite{BliznetsFPP14,BliznetsFPP14a,DrangeDLS15,DrangeFPV13,DrangeP14,FominV13,KaplanST99a,VillangerHPT09} for a broader discussion and pointers to relevant literature.

\minimumfillin is NP-hard, as shown by Yannakakis~\cite{yannakakis}, however the reduction showing this is quite non-obvious; the complexity status of the problem was among the 12 open problems at the end of the first edition of the Garey and Johnson's book~\cite{garey1979computers}. The study of \minimumfillin from the point of view of the parameterized complexity, with $k$ being the obvious parameter of interest, started with the pioneering work of Kaplan et al.~\cite{KaplanST99a}. They proposed a fixed-parameter (FPT) algorithm with running time $\Oh(16^kk^6+k^2mn)$ that is based on locating holes in the graph and branching on possible ways of adding edges to get rid of them. A similar strategy worked also for \properintervalcompletion~\cite{KaplanST99a}, the problem of adding as few edges as possible to obtain a proper interval graph, but not for \intervalcompletion. The fixed-parameter tractability of the latter has been resolved by Villanger et al.~\cite{VillangerHPT09} only several years thereafter.

A complete turning point came four years ago, when Fomin and Villanger~\cite{FominV13} presented an algorithm for \minimumfillin with {\em{subexponential}} parameterized complexity, more precisely with running time $\Oh(2^{\Oh(\sqrt{k}\log k)}+k^2nm)$. This was an immense surprise to the parameterized complexity community, since subexponential parameterized algorithms, i.e., with running time $2^{o(k)}\cdot n^{\Oh(1)}$, were known essentially only in two restricted settings: in topologically constrained graphs via the technique of bidimensionality (see e.g.~\cite{DemaineH08}), and in tournaments, with an important example of \fast~\cite{AlonLS09,faster-feige,faster-karpinski}. For most natural parameterized problems, the existence of such algorithms can be excluded under the Exponential Time Hypothesis, which (when combined with the {\em{Sparsification Lemma}}~\cite{sparsification}) essentially states that there is no algorithm for \sat[3] with running time $2^{o(n+m)}$; cf.~\cite{our-book,grohe:book}.

The work of Fomin and Villanger~\cite{FominV13} presented a conceptual breakthrough in the approach to completion problems for subclasses of chordal graphs. The main idea is not to focus on breaking individual obstacles such as holes by single edge additions, as was proposed in the previous works, but to build a structural decomposition of the completed graph by means of a dynamic programming algorithm that minimizes the number of edges missing in the decomposition. The crux is to show that this dynamic programming can be restricted to a space of states that has size $2^{\Oh(\sqrt{k}\log k)}\cdot n^{\Oh(1)}$, and can be enumerated efficiently. In case of \minimumfillin, the considered decomposition is the clique tree, but this generic approach can be in principle applied to every subclass of chordal graphs whose graphs have a well-defined global structure. And so, following~\cite{FominV13}, subexponential parameterized algorithms have been designed for \thresholdcompletion~\cite{DrangeFPV14}, \triviallyperfectcompletion~\cite{DrangeFPV14}, \properintervalcompletion~\cite{BliznetsFPP14}, and even 
\intervalcompletion~\cite{BliznetsFPP14a}. Apart from \properintervalcompletion, for all these problems the algorithms have running time $2^{\Oh(\sqrt{k}\log k)}\cdot n^{\Oh(1)}$; for the former we are currently stuck at exponent $\Oh(k^{2/3}\log k)$, but this is conjectured to be an artifact of the technique~\cite{BliznetsFPP14}. Let us remark that in each of these cases the idea of Fomin and Villanger only provides the basic outline of the strategy, while the actual implementation is always class-specific and requires involved technical insight.

Drange et al.~\cite{DrangeFPV13,DrangeFPV14} complemented their results with a number of lower bounds suggesting that even a slight deviation from the setting of adding edges to a subclass of chordal graphs leads to the non-existence of a subexponential parameterized algorithm under the assumption of ETH. Thus, the aforementioned problems for which such algorithms exist are in fact intriguing ``singularity points'' on the complexity landscape. It is interesting that these singularity points actually correspond to problems that have important practical applications. Hence, the study of this phenomenon is an important direction that naturally belongs to the {\em{optimality programme}} (cf.~\cite{eth-survey,Marx12}): a trend in parameterized complexity that focuses on systematic study of parameterized problems by providing possibly tight upper and lower bounds on their complexity.

From this point of view, the first natural question is exactly how ``deep'' are the aforementioned singularity points. More precisely, is the running time of the form $2^{\Oh(\sqrt{k}\log k)}\cdot n^{\Oh(1)}$ optimal, say under ETH? For subexponential algorithms derived using bidimensionality, the non-existence of $2^{o(\sqrt{k})}\cdot n^{\Oh(1)}$-time algorithms under ETH usually follows already from the known NP-hardness reductions~\cite{our-book,grohe:book}. However, for the considered completion problems this is not the case: as Fomin and Villanger already observed in~\cite{FominV13}, the classic NP-hardness reduction for \minimumfillin gives only a $2^{o(k^{1/6})}\cdot n^{\Oh(1)}$ lower bound. A similar situation holds for \fast, for which the fastest known algorithms work in time $2^{\Oh(\sqrt{k})}\cdot n^{\Oh(1)}$~\cite{faster-feige,faster-karpinski}. 

For this reason, the question of establishing tight upper and lower bounds for \minimumfillin was already asked explicitly by Fomin and Villanger~\cite{FominV13}, repeated by Marx in his survey on the optimality programme~\cite{Marx12}, and then reiterated for respective subclasses in all the works~\cite{BliznetsFPP14,BliznetsFPP14a,DrangeFPV13}. The goal of this paper is to remedy this situation by providing complexity foundations for proving that the square root in the exponent of the running time is hard to improve.

\paragraph*{Our results.} First, we investigate how strong lower bounds for completion problems can be obtained when we assume only ETH.

\begin{theorem}\label{thm:only-eth}
Unless ETH fails, there is an integer $c\geq 1$ such that there are no $2^{\Oh(\sqrt{n}/\log^c n)}$, and consequently no $2^{\Oh(k^{1/4}/\log^c k)}\cdot n^{\Oh(1)}$ algorithms for the following problems: \minimumfillin, \intervalcompletion, \properintervalcompletion, \triviallyperfectcompletion, \thresholdcompletion, {\sc{Chain Completion}}.
\end{theorem}

Here, a graph $G$ is a {\em{chain graph}} if it is a bipartite graph with a fixed bipartition $A\uplus B$, where the vertices of $A$ can be ordered by a linear ordering $\preceq$ such that $N(u)\subseteq N(v)$ whenever $u\preceq v$, for all $u,v\in A$. In the {\sc{Chain Completion}} problem we are given a bipartite graph with a fixed bipartition $A\uplus B$, and we want to add at most $k$ edges between $A$ and $B$ to obtain a chain graph. {\sc{Chain Completion}} also admits an algorithm with running time $\Oh(2^{\Oh(\sqrt{k}\log k)}+k^2nm)$~\cite{FominV13}.

Unfortunately, Theorem~\ref{thm:only-eth} does not provide a tight result for completion problems we are interested in: we still have a gap between $k^{1/4}$ and $k^{1/2}$. We in fact prove a higher lower bound, but we need to ground it on a stronger complexity assumption. More precisely, we consider the approximability of the \minb problem on $d$-regular graphs: given a $d$-regular graph $G$ with an even number of vertices, find a partition of $V(G)$ into {\em{equal-sized}} parts that minimizes the number of edges between the parts. For $0 \leq \bounda < \boundb \leq 1$, problem \gapminbd[d]$_{[\alpha, \beta]}$ is defined as the gap problem of deciding whether the optimum solution is at most $\bounda m$, or at least $\boundb m$.

\begin{hypothesis} 
\label{h.bisection} 
There exist $0 \leq \bounda < \boundb \leq 1$, and an integer $d > \frac{4}{\boundb-\bounda}$, such that there is no $2^{o(n)}$-time algorithm for \gapminbd[d]$_{[\alpha, \beta]}$. 
\end{hypothesis} 

\todo{If you think this is needed, combine above definition with \gapminbd[d]$_{[\bounda, \boundb]}$}

\begin{theorem}\label{thm:main}
Unless Hypothesis~\ref{h.bisection} fails, there is no $2^{o(n+m)}$-time algorithm for {\sc{Chain Completion}}, and no $2^{o(n)}$-time algorithms for \minimumfillin, \intervalcompletion, \properintervalcompletion, \triviallyperfectcompletion, and \thresholdcompletion. Consequently, none of these problems can be solved in time $2^{o(\sqrt{k})}\cdot n^{\Oh(1)}$.
\end{theorem}

Thus, Theorem~\ref{thm:main} asserts that the existence of a substantially faster algorithm for any of the considered completion problems would lead to a breakthrough in the design of approximation algorithms for \minb. We do not intend to take a stance on whether Hypothesis~\ref{h.bisection} is true or false. However, the existence of an algorithm refuted by Theorem~\ref{h.bisection} is far from the current knowledge: the best polynomial-time approximation algorithms for \minb achieve approximation factor $\Oh(\log OPT)$~\cite{Racke08}, whereas to refute Hypothesis~\ref{h.bisection} one would need to obtain approximation factor arbitrarily close to one. This could be, however, possible due to the assumption of $d$-regularity, for a constant $d$, as well as the access to subexponential-time computations. We are not aware of any approximation algorithms for \minb that would significantly use any of these assumptions. In fact, Berman and Karpinski~\cite{berman-karpinski} have shown that approximating \minb in graphs of maximum degree $3$ is as hard as in general graphs from the point of view of polynomial time approximation.

Finally, using the methodology of Theorem~\ref{thm:only-eth} we can give an improved lower bound for \fast. For this problem, the first subexponential FPT algorithm with running time $2^{\Oh(\sqrt{k}\log k)}\cdot n^{\Oh(1)}$ was given by Alon et al.~\cite{AlonLS09}, which was later improved to $2^{\Oh(\sqrt{k})}\cdot n^{\Oh(1)}$ independently by Feige~\cite{faster-feige} and by Karpinski and Schudy~\cite{faster-karpinski}.

\begin{theorem}\label{thm:only-eth-fast}\footnote{As the reduction in the proof of Theorem~\ref{thm:only-eth-fast} is randomized, here we assume the randomized version of the ETH.}
Unless ETH fails, there is an integer $c\geq 1$ such that there is no $2^{\Oh(\sqrt{n}/\log^c n)}$, and consequently no $2^{\Oh(k^{1/4}/\log^c k)}\cdot n^{\Oh(1)}$-time algorithm for \fast.
\end{theorem}

A direct inspection of the classic NP-hardness proofs~\cite{AilonCN08,Alon06,CharbitTY07} gives only a $2^{\Oh(n^{1/4})}$, and consequently also a $2^{\Oh(k^{1/8})}\cdot n^{\Oh(1)}$ lower bound, but Theorem~\ref{thm:only-eth-fast} is still not tight: we have a gap between $k^{1/2}$ and $k^{1/4}$, similarly as in Theorem~\ref{thm:only-eth}. Unfortunately, despite efforts we are unable to close this gap using Hypothesis~\ref{h.bisection}, or a similar assumption. 

\paragraph*{Our techniques.} The main idea of this work is to carefully combine the standard approach of starting with a sparse instance of \sat[3], prepared using the Sparsification Lemma~\cite{sparsification}, with the existence of almost linear PCPs. More precisely, an application of the Sparsification Lemma reduces solving an instance $\varphi$ of \sat[3] to solving a number of instances $\phi_1,\phi_2,\ldots,\phi_\ell$ of \sat[3] on the same variable sets, but in each $\phi_i$ we have that the number of clauses $m$ is linear in the number of variables $n$. Then, to each $\phi_i$ we apply the result of Dinur~\cite{Dinur07} on the existence of almost linear-size PCPs. This transforms each $\phi_i$ into an instance $\phi'_i$ of \sat[3] where the number of variables and clauses is $\Oh(n \log^c n)$, but the instance is either satisfiable, or at most a $\rho$-fraction of clauses can be satisfied, for some $\rho<1$. 

Having such an instance at hand, we can proceed with the chain of reductions proposed by Garey et al.~\cite{GareyJS76} and Yannakakis~\cite{yannakakis}. Namely, we first use simple manipulations to obtain an equivalent instance of \maxcut with a gap, and then reduce it to the \ola (\olashort) problem: Given a graph $G$ on $n$ vertices, find a linear ordering $\pi\colon V(G)\to \{1,2,\ldots,n\}$ that minimizes the cost defined as $\sum_{uv\in E(G)}|\pi(u)-\pi(v)|$; value $|\pi(u)-\pi(v)|$ is also called the {\em{cost}} of $uv$. 
This is precisely the moment where we exploit that we are working with an instance of \maxcut with a gap. 
Namely, the construction introduces a huge clique to the complement of the \maxcut instance; this clique is supposed to separate in the ordering the sides of an optimum max-cut solution. In Garey et al.~\cite{GareyJS76}, its size must be large enough so that the cost of any edge jumping over the clique dwarfs the ``noise'' contribution that is given by internal ordering of parts on the left and on the right. Hence, the clique is chosen to be of size $\Theta(n^4)$, which explodes the size of the instance. However, we observe that by starting with an instance with a gap, we can accommodate the noise in the gap, and therefore we only need the clique to be of linear size. Thus, we obtain the following theorem that can be of independent interest.

\begin{theorem}\label{thm:eth-ola}
Unless ETH fails, for some $c > 1$ there is no algorithm solving \ola in time $2^{\Oh(n/\log^c n)}$.
\end{theorem}

Note that \olashort can be solved by a simple Held-Karp dynamic programming on subsets in time $2^n\cdot n^{\Oh(1)}$, so Theorem~\ref{thm:eth-ola} gives an almost tight result.


To obtain Theorem~\ref{thm:only-eth} one just applies slightly modified elegant reduction from (the decision variant of) \olashort to \minimumfillin, proposed by Yannakakis~\cite{yannakakis}. The reduction firstly reduce \olashort to \chaincompletion. Because of this fact later we get a very specific \minimumfillin output instance. And this immediately gives us a reduction to \properintervalcompletion and \intervalcompletion problems.  With slight changes similar result is obtained for \thresholdcompletion and \triviallyperfectcompletion.  
Theorem~\ref{thm:only-eth-fast} is derived by applying the same methodology to the randomized hardness reduction for \fast of Ailon et al.~\cite{AilonCN08}.

The moment when we lose tightness in our chain of reductions is precisely the last reduction from \olashort to \minimumfillin. This is because in the construction we introduce a new vertex for every {\em{edge}} of the original graph, which blows up the vertex set quadratically. Therefore, the reduction would give a tight lower bound if we started with a sparse instance of \ola. This is precisely our idea behind the proof of Theorem~\ref{thm:main}.

Namely, we observe that by starting from the hardness of \gapminbd[d]$_{[\alpha, \beta]}$, instead of \maxcut, we can give an alternative hardness reduction to \ola in graphs
of bounded degree by a constant $d$ (\oladshort[d] for short). This reduction is much more intricate. Essentially, we replace the usage of a huge clique (which would blow up the degree) by a careful construction using several layers of expander graphs that mimics the same behaviour. This proves the following result.

\begin{theorem}\label{thm:hyp-ola}
Unless Hypothesis~\ref{h.bisection} fails, there exists an integer $d \in \N$ such that there is no $2^{o(n)}$-time algorithm solving \oladshort[d] in multigraphs.
\end{theorem}

\todo{If you think that is reasonable, combine \olashort definitions with \oladshort[d]}

From this, Theorem~\ref{thm:main} follows similarly as Theorem~\ref{thm:only-eth} followed from Theorem~\ref{thm:eth-ola}.

\paragraph*{Related work.} The combination of ETH, Sparsification Lemma, and PCP tools was already used in the recent line of work on developing the theory of hardness of FPT approximation; see e.g.~\cite{BonnetE0P15,BonnetLP15,BonnetP14,BonnetP14b,HajiaghayiKK13,KhotS15}. There, the main goal is to provide lower bounds on how the access to FPT-time computations can help in the design of approximation algorithms for fundamental inapproximable problems, like {\sc{Clique}} or {\sc{Set Cover}}. 

The existence of almost linear PCPs combined with ETH was also used by Marx~\cite{Marx07a} to give almost tight lower bounds for the running times of polynomial time approximation schemes (PTASes) for several fundamental problems on planar graphs and in the Euclidean plane. The basic principle behind the approach of Marx~\cite{Marx07a} is very similar to ours. Namely, it is observed that using ETH and PCP tools one can prepare a sparse instance of the maximization variant of \sat[3] where it is hard to distinguish between a fully satisfiable instance and an instance where only a $(1-\epsilon)$-fraction of the clauses can be satisfied in time $2^{\Oh(n/\log^c n)}$, for some constant $c$ (see also Theorem~\ref{thm:hard-gapsat} in this work). This observation is used as a base for further reductions refuting the existence of certain approximation algorithms for geometric and planar problems.

Thus, this work provides another example where the said combination appears to be useful. This time we use it to prove improved lower bounds on the complexity of FPT algorithms solving certain graph modification problems exactly. The new idea in this work is that the gap property of the considered instances can be used not only to exclude the existence of approximation algorithms, but also to limit the instance size explosion in a chain of NP-hardness reductions by using more thrifty constructions.

Hypothesis~\ref{h.bisection} can possibly have links with the hypothesis put forward by Feige~\cite{Feige02}. Essentially, Feige conjectures the hardness of distinguishing a ``typical'' \sat[3] instance from a satisfiable one in polynomial time. From this, he derives as a corollary a variant of Hypothesis~\ref{h.bisection}, but without the assumption of $d$-regularity, for a constant $d$, and with subexponential time replaced by polynomial. It is conceivable that Hypothesis~\ref{h.bisection} can be also implied by some stronger variant of Feige's conjecture, but we refrain from formalizing this link due to many technical problems that arise when attempting to do this.

\paragraph*{Outline of the paper.} 

In Section~\ref{sec:preliminaries}, we give all the needed definitions and list the tools used in our reductions,
such as the Exponential Time Hypothesis and the PCP theorem.
Next, in Section~\ref{sec:ola-eth}, we prove Theorem~\ref{thm:eth-ola}.
The main technical contribution of the paper, that is the proof of Theorem~\ref{thm:hyp-ola},
is contained in Section~\ref{sec:sparse}. Theorem~\ref{thm:hyp-ola} is the main ingredient in the proof of Theorem~\ref{thm:main},
that is presented in Section~\ref{sec:completions}. Section~\ref{sec:fast} is devoted to
the proof of Theorem~\ref{thm:only-eth-fast}.

%% file: sec_prelim.tex
\section{Preliminaries}
\label{sec:preliminaries}

\subsection{Parameterized complexity}
\label{s.parameterized}

\term{A parameterized problem} $Q$ is a subset of $\Sigma^{*} \times \mathbb{N}$, for a fixed finite alphabet $\Sigma$. 
An instance of~the~problem $Q$ is an element $(x,k) \in \Sigma^{*} \times \mathbb{N}$ with the integer $k$ called the parameter. 
A parameterized problem is \term{fixed parameter tractable} if there exists an algorithm deciding whether $(x,k) \in Q$, and
working in time $f(k) \cdot \text{poly}(|x|)$ for every instance $(x,k)$, where $f$ is a computable function.
For denoting the time of~an~algorithm we will use the $\mathcal{O}^{*}$ notation that suppresses factors polynomial in the size of the input, 
e.g. $\mathcal{O}^{*}(f(k))$.

We denote a polynomial deterministic linear reduction from a problem $X$ to $Y$ by $X \linleq Y$,
i.e., $X \linleq Y$ means that there is a polynomial time deterministic algorithm, which
given an instance $I$ of the problem $X$ produces an instance $I'$ of the problem $Y$ of size $\Oh(|I|)$,
such that $I$ is a yes-instance if and only if $I'$ is a yes-instance.

\subsection{Graph notations}
\label{s.graphs}

A standard graph theoretical notation is used throughout the paper. 
\term{A graph} $G$ is a tuple $(\vertices[G], \edges[G])$, where $\vertices[G]$ is the set of \term{vertices} and $\edges[G] \subseteq {\vertices[G] \choose 2}$ is the set of \term{edges}. 
When it is clear from the context what graph we refer to, we will denote by $V$ and $E$ the set of vertices and edges of $G$. 
A graph $H$ is \term{a subgraph of $G$} if $\vertices[H] \ss \vertices[G]$ and $\edges[H] \ss \edges[G]$. 
A graph $H$ is \term{an induced subgraph of $G$} if $\vertices[H] \ss \vertices[G]$ and $\edges[H] = \edges[G] \cap {\vertices[H] \choose 2}$. 
An induced subgraph of $G$ with vertex set $X$ is denoted by $G[X]$.
\term{A complement} of $G$ is the graph with the vertices $V$ and the edge ${V \choose 2} - E$ and we denote it by $\overline G$.
For $X \subseteq V$, we denote by $\delta_G(X)$ the set of edges with exactly one endpoint in $X$.
We define \term{a cut}, as the set of edges $E_G(A,B)$, for a partition $(A,B)$ of $V$. We denote the size of the cut by $|E_G(A,B)|$.

We often use $n$, and $m$ to denote the size of $V$ and $E$, respectively. 
For a vertex $v$, $\deg_G(v)$ denotes the degree (the number of incident edges) of the vertex $v$. 
We say that a graph is \term{$d$-regular} if the degree of each vertex is equal to $d$.
We denote by $\Delta_G$ the maximum degree of $G$. 
The set $N_G(v) = \{ w : (v, w) \in \edges[G] \}$ is the neighbourhood of $v$. 
We extend this notation to subsets of vertices $X$, i.e. $N_G(X) = \bigcup_{v \in X} N_G(v) \setminus X$.
We will also omit subscripts, i.e. $\deg(v), \Delta, N(X), \delta(X), E(U,V)$, when it is clear from the context, which graph we refer to.

If $X, Y \ss V$ are disjoint, then $E_G(X, Y)$ is the set of edges between $X$ and $Y$. 
We use $G[X, Y]$ to denote \term{the induced bipartite subgraph of $G$ with parts $X$ and $Y$.}
That is, $G[X, Y]$ is the graph with vertex set $X \cup Y$ that contains precisely the edges $E_G(X, Y)$. 

For $X \subseteq V$, we denote by $\delta_G(X)$ the set of edges with exactly one endpoint in $X$.
We define \term{a cut}, as the set of edges $E_G(A,B)$, for a partition $(A,B)$ of $V$. We denote the size of the cut by $|E_G(A,B)|$.

\subsection{Expanders}
\label{s.expanders} 

\term{The Cheeger number $h(G)$ of a graph $G$} is defined as
$$h(G) := \min \left\{ \frac{|\delta (X)|}{|X|} : X \subseteq V(G), |X| \leq \frac{|V(G)|}{2} \right\},$$
We say that a graph $G$ is \term{a $(d, e)$-expander} if it is $d$-regular and has Cheeger number of at least $e$. 
When discussing expanders, it is convenient to allow parallel edges and self-loops,
which naturally appear in most of the expanders constructions.
It is important to note that a self-loop contributes $1$ to the degree of a vertex.

There are many efficient constructions of expanders available in the literature. 
The following theorem produces graphs with provably optimal Cheeger numbers: 
\begin{theorem} \cite{Lubotzky88, Morgenstern94}
\label{t.expander} 
Let $d = p^k + 1$, where $p$ is a prime and $k \in \N$, and let $q$ be a prime congruent to $1 \text{ mod } 4$.
Assume $p \neq q$.  
There exists a $(d, \half{d} - \sqrt{d - 1})$-expander on $q + 1$ vertices.
Furthermore, such an expander can be constructed in polynomial time. 
\end{theorem}

Throughout this paper, we use expanders of various sizes, hence
the following theorem appears to be useful in our setting,
despite providing weaker bound on the Cheeger number.

\begin{theorem}\label{thm:expander-factory}[Theorem 21.19 of \cite{arora-barak}]
Let $p>0$ be a real. Then there exists a positive integer $d$, such that for every positive integer $n$ there exists a $d$-regular multigraph $G_{n,d}$ on $n$ vertices with $h(G_{n,d})\geq p$. Moreover, graph $G_{n,d}$ can be constructed in time polynomial in $n$.
\end{theorem}

\subsection{Linear arrangements}

\term{A linear arrangement} of a graph $G=(V,E)$ is a function (a vertex ordering) $\aorder : V \rightarrow \{1, \ldots, n\}$. \term{The cost} of a linear arrangement $\aorder$
is defined by $\sum_{uv \in E} |\aorder(u) - \aorder(v)|$. We call $\aorder$ the optimum linear arrangement if its cost is minimized over all vertex orderings of $G$
and we denote this cost by $\olashort(G)$. 

\subsection{Satisfiability}

We employ a standard notation related to \justsat\ problems. 
We use symbols $x_1, \ldots ,x_n$ for the variables of an instance, and $C_1, \ldots, C_m$ for the clauses. 
An $l$-CNF formula is E$l$-CNF, if it has exactly $l$ literals. We say that an assignment of $x_1,...,x_n$ NAE-satisfies an $l$-CNF formula if every clause contains 
a satisfied and an unsatisfied literal.
An $l$-AND is a conjunction of clauses, where each clause is a conjunction of at most $l$ literals.

\subsection{Exponential Time Hypothesis}

The Exponential Time Hypothesis (ETH), introduced by Impagliazzo, Paturi and Zane~\cite{sparsification, ImpagliazzoP01}
is now an established tool used for proving conditional lower bounds in the parameterized complexity area (see~\cite{eth-survey}
for a survey on ETH-based lower bounds). Intuitively, ETH states that \sat[3] cannot be solved in time subexponential in the number of variables. 


\begin{hypothesis}[Exponential Time Hypothesis (ETH) \cite{sparsification,ImpagliazzoP01}]
There is no $2^{o(n)}$ time randomized algorithm for \sat[3]. 
\end{hypothesis}

\begin{lemma} \label{sparsification} (Sparsification Lemma, \cite{sparsification})
  For every $\epsilon > 0$, there is an algorithm that takes a $3$-CNF formula $\phi$ and returns $l$ $3$-CNF formulas $\phi_1, ..., \phi_l$, such that: i) $l = \mathcal{O}(2^{\epsilon n})$, 
  ii) for every $i$, $\phi_i$ has $n$ variables, and every such variable appears in at most $c_{\epsilon}$ clauses of $\phi_i$, for some constant $c_{\epsilon} \geq 0$,
  iii) $\phi$ is satisfiable if and only if at least one of $\phi_i$ is satisfiable. The running time of the algorithm is $\mathcal{O}^{*}(2^{\epsilon n})$.
  \end{lemma}

Consequently, based on ETH we have subexponential hardness of \sat[3]
in terms of both the number of variables and clauses.

\begin{theorem}[\cite{sparsification}]
Unless ETH fails, there is no $2^{o(m + n)}$ time algorithm for \sat[3].
\end{theorem}

\subsection{Gap problems and PCPs}

In a gap version of a problem, the input instance is promised
to belong to one of two languages specifying the allowed input, and the goal is to decide
which case (language) a given instance belongs to.
Gap problems are associated with two parameters $\bounda, \boundb$.
For example, in $\gapmaxcut_{[\bounda,\boundb]}$, we are to distinguish
between the case when a given graph admits a cut of size at least $\boundb m$
and the case when a given graph does not admit a cut of size larger than $\bounda m$.
Similarly for satisfiability problems with gap $[\bounda,\boundb]$
we are to distinguish between a formula, for which at least $\boundb m$ clauses can be satisfied,
and the case where it is impossible to satisfy more than $\bounda m$ clauses.

To introduce gaps in our reductions, we will use the following fundamental result of Dinur.

\todo{give some reference for PCP notation}

\begin{theorem} \label{linear-pcp} (Almost Linear Size PCP, \cite{Dinur07}) 
\sat[3] $\in \textbf{PCP}_{1, \frac{1}{2}} (\log(n) + \mathcal{O}(\log \log(n)), \mathcal{O}(1))$
\end{theorem}

%% file: sec_nosubexp.tex
\section{Combining known reductions for \olashort with PCPs and ETH}
\label{sec:ola-eth}


In this section we discuss the proof of Theorem~\ref{thm:eth-ola},
i.e., explain how to show that unless ETH fails, 
  there is no $2^{\mathcal{O}(\frac{n}{\log^c(n)})}$-time algorithm for \ola, for some constant $c \in \N$.
We combine several well known concepts from the complexity theory: the ETH hypothesis, the PCP theorem, the Sparsification Lemma, gap reductions, and gap amplification. 

First, we show that the combination of the ETH hypothesis,
together with the almost linear size PCP theorem, and the Sparsification Lemma,
gives subexponential hardness of \gapesat[3]. This fact was already observed and used by Marx~\cite{Marx07a}.

\begin{theorem}[see also Lemma 2.2 of~\cite{Marx07a}]
\label{thm:hard-gapsat}
Unless ETH fails, there exist $c \in \N$, and $r \in (0,1)$ such that there is no $2^{\mathcal{O}(\frac{m}{\log^c(m)})}$-time algorithm for \gapesat[3]$_{[r, 1]}$.
\end{theorem}

\begin{proof} 
 From Theorem \ref{linear-pcp}, it follows that there is a $\textbf{PCP}_{1, \frac{1}{2}} (\log(n) + A \log \log(n), B)$ verifier $V$ for $3$-SAT, for some constants $A,B > 0$.\\
 
 Assume we are given the input E3-CNF formula $\phi$. Fix any $\epsilon >0$ and apply Lemma \ref{sparsification} to $\phi$, obtaining $l = \mathcal{O}(2^{\epsilon n})$ instances $\phi_i$, $i=1,...,l$,
 each with size bounded by $cn$, for some constant $c > 0$, depending on $\epsilon$. Let $R = \log(c n) + A \log \log(c n)$. \\
 
 If we take $\phi_i$ as the input to the verifier $V$, then for each random binary string of length at most $R$, it reads at most $B$ bits of the proof
 (we will assume without loss of generality that $V$ always reads exactly $B$ bits).
 Thus, $V$ has access to $N$ bits of the proof in total, where $N \leq B \cdot 2^{R} = B \cdot c n \log^A(c n).$
 We create variables $x_1, ... , x_N$ corresponding to each of these bits. For a boolean string $r$ of length $R$, we define a boolean function $f_r : \{0,1\}^B \rightarrow \{0,1\}$, 
 with its arguments corresponding to relevant variables $x_{i_1}, ... ,x_{i_B}$ that can be read from the proof. $f_r$ is defined as follows, for every $S \in \{0,1\}^B$
 it evaluates to true on $S$ if and only if $V$ accepts $\phi_i$ using $S$ as the bits from the proof.
 Every $f_r$ can be written as an equivalent E$3$-CNF formula $F_r$ (possibly adding some constant number of variables). 
 Let $C$ be an upper bound (depending on $B$) for the number of clauses for each such formula. Let $\phi'_i$ be a formula combined by AND 
 of all $F_r$. It has the number of clauses bounded by $C \cdot 2^R$ and the following holds.  \\
 
 If $\phi_i$ is satisfiable, then there is a proof for which $V$ always accepts.
 By taking the corresponding values from this proof and setting them to $x_1, ... ,x_N$,
 we obtain a satisfying assignment of $\phi'$.\\
 
 For the proof in the other direction, if $\phi_i$ is not satisfiable, 
 then for any proof for at least half of all binary strings $r$ the verifier $V$ rejects $\phi_i$. 
 Consider an arbitrary proof and set the variables $x_i$ to the corresponding values from this proof. 
 Observe, that for the binary strings $r$ for which $V$ rejects $\phi_i$
 at least one clause in $F_r$ at least one clause is false.
 Consequently, at least $\frac{1}{2} \cdot 2^R$ clauses of $\phi_i'$ are false,
 i.e., a fraction of at least $\frac{1}{2C}$ of the total number of clauses of $\phi'_i$.\\ 
 
 Finally, let us put $c = A+1$, and $r = 1 - \frac{1}{2C}$ and suppose there is a $\mathcal{O}^{*}(2^{\mathcal{O}(\frac{m}{\log^{A+1}(m)})})$-time algorithm for \gapesat[3]$_{[1-\frac{1}{2C}, 1]}$.
 Using this algorithm, we can check whether $\phi_i$ is satisfiable in time $2^{o(n)}$ and 
 therefore check whether the initial formula $\phi$ is satisfiable in time $\mathcal{O}^{*}(2^{\epsilon n} + 2^{\epsilon n} \cdot 2^{o(n)}) = \mathcal{O}^{*}(2^{\epsilon n})$. 
 As $\epsilon$ is an arbitrary positive number, we obtain a contradiction with the ETH.
\end{proof}

Next, we inspect the chain of three textbook NP-hardness reductions~\cite{schaefer1978complexity}, that starts with \gapesat[3] and ends with \gapmaxcut. 
All of them produce an output instance of linear size in terms of the size of the input instance and at the same time preserve gaps.
Combining them with the previous theorem, we will establish (almost) subexponential hardness of \gapmaxcut.

\begin{theorem}
\label{thm:hard-gapmaxcut}
 $$\gapesat[3]_{[\bounda, \boundb]} \linleq \gapenaesat[4]_{[\bounda, \boundb]} \linleq $$
 $$\gapenaesat[3]_{[\frac{1 + \bounda}{2}, \frac{1 + \boundb}{2}]} \linleq \gapmaxcut_{[\frac{16 + \bounda}{18}, \frac{16 + \boundb}{18}]}$$
\end{theorem}

\begin{proof}

For simplicity, the proof of Theorem~\ref{thm:hard-gapmaxcut} is split 
into the following four separate lemmas.

\begin{lemma}\label{lem:3sat-nae4sat}
 $$\gapesat[3]_{[\bounda, \boundb]} \linleq \gapenaesat[4]_{[\bounda, \boundb]}$$
\end{lemma}

\begin{proof} 
Given an E3-CNF formula $\phi = C_1 \wedge \ldots \wedge C_m$, 
we create a new E4-CNF formula $\phi' = C_1' \wedge \ldots \wedge C_m'$
by adding a new variable $z$ to every clause.
That is, if $C_i = l_1 \vee l_2 \vee l_3$, 
then we set $C'_i = l_1 \vee l_2 \vee l_3 \vee z$.

If some assignment satisfies $k$ clauses of $\phi$, then by additionally setting $z = 0$
the corresponding $k$ clauses of $\phi'$ are NAE-satisfied.
In the other direction, if some assignment $\varphi$ NAE-satisfies $k$ clauses of $\phi'$, then 
its negation $\tilde{\varphi}$ also NAE-satisfies $k$ clauses of $\phi'$.
W.l.o.g. assume that $\varphi$ sets $z$ to false,
which means that $\varphi$ restricted to the variables of $\phi$ satisfies $k$ clauses of $\phi$.
Consequently we obtain a gap preserving reduction and the theorem follows.
\end{proof}

\begin{lemma}\label{lem:nae4sat-nae3sat}
 $$\gapenaesat[4]_{[\bounda, \boundb]} \linleq \gapenaesat[3]_{[\frac{1 + \bounda}{2}, \frac{1 + \boundb}{2}]}$$
\end{lemma}

\begin{proof}
 
Given an E4-CNF formula $\phi = C_1 \wedge ... \wedge C_m$, we add $m$ new variables $z_1, ... ,z_m$, and replace every clause $C_i = l_1 \vee l_2 \vee l_3 \vee l_4$
with the following two clauses: $C'_i = l_1 \vee l_2 \vee z_i$, $C''_i = l_3 \vee l_4 \vee \neg z_i$, obtaining an E3-CNF formula $\phi'$. 
In the following we show that one can NAE-satisfy at least $k$ clauses of $\phi$
if and only if one can NAE-satisfy at least $m+k$ clauses of $\phi'$,
which suffices to prove the theorem.

First, observe that if an assignment NAE-satisfies $C_i$, then by setting
$z_i$ appropriately both $C'_i$ and $C''_i$ become NAE-satisfied.
On the other hand if an assignment does not NAE-satisfy $C_i$,
then setting $z_i$ to an arbitrary value
NAE-satisfies exactly one clause out of $C_i'$ and $C_i''$.
Consequently if there is an assignment, which NAE-satisfies $k$ clauses of $\phi$,
then it can be extended to an assignment of variables of $\phi'$, which NAE-satisfies
$m+k$ variables of $\phi'$.

Let us assume that there is an assignment $\varphi$, which NAE-satisfies at least $m+k$ clauses of $\phi'$.
Note that there is a set $I$ of at least $k$ indices $i$, 
such that $\varphi$ NAE-satisfies both $C_i'$ and $C_i''$.
Consider a fixed $i \in I$ and w.l.o.g. assume that $\varphi$ sets $z_i$ to false.
As $\varphi$ NAE-satisfies both $C_i'$ and $C_i''$, we infer that 
$\varphi$ sets at least one of the literals $l_1$, $l_2$ to true,
and at least one of the literals $l_3, l_4$ to false.
Consequently $\varphi$ restricted to the variables of $\phi$ NAE-satisfies
all the clauses $C_i$ for $i \in I$, which finishes the proof of the theorem as $|I| \ge k$.
\end{proof}

\begin{lemma}\label{lem:nae3sat-mmaxcut}
$$ \gapenaesat[3]_{[\bounda, \boundb]} \linleq \gapmmaxcut_{[\frac{3 + 2 \bounda}{6}, \frac{3 + 2 \boundb}{6}]}$$
\end{lemma}

\begin{proof}
 Given an E3-CNF formula $\phi = C_1 \wedge ... \wedge C_m$, let $n_i$ be the number of 
 occurrences of $x_i$ in $\phi$, both in the positive and negative form.
 We construct a multigraph $G$ with $2n$ vertices, and $6m$ edges as follows.
 For each variable $x_i$, we create two vertices corresponding to literals $x_i$, $\neg x_i$, and add exactly $n_i$ edges between them. 
 Finally, for every clause $C_i$ we add a triangle connecting
 the vertices corresponding to its literals.
 We will prove that there is an assignment which NAE-satisfies at least $k$ clauses of $\phi$ 
 if and only if $G$ admits a cut of size at least $3m+2k$.

 Let us assume that there is an assignment which NAE-satisfies $k$ clauses of $\phi$.
 We put vertices of $G$, that correspond to literals evaluated to true,
 on one side of the cut, and remaining literals to the other side of the cut.
 Every edge connecting $x_i$ and $\neg x_i$ is clearly in the cut, 
 and all of them contribute $\sum_{i} n_i = 3m$ to the cut. 
 Furthermore, every triangle corresponding to a NAE-satisfied clause has exactly two edges cut, all of them contribute $2k$ to the cut. 
 Thus, we obtain the cut of size at least $3m+2k$.
 
 Let us assume that $G$ has a cut of size at least $3m+2k$. First, suppose that $x_i$ and $\neg x_i$ are on the same side of the cut for some variable $x_i$.
 They contribute at most $n_i$ edges to the cut,
 and we can move one of the vertices $x_i$ or $\neg x_i$ to the other side of the cut
 without decreasing the number of edges in the cut.
 Thus, we can assume that variables are separated from their negations. The edges connecting $x_i$ and $\neg x_i$ contribute exactly $3m$ to the cut. 
 If a clause triangle is cut $2$ times, then it corresponds to a NAE-satisfied clause,
 otherwise the clause is not NAE-satisfied as all its literals have the same value.
 Thus, we deduce that the considered cut corresponds to an assignment
 which NAE-satisfies $k$ clauses of $\phi$.
\end{proof}

\begin{lemma}\label{lem:mmaxcut-maxcut}
$$\gapmmaxcut_{[\bounda, \boundb]} \linleq \gapmaxcut_{[\frac{2 + \bounda}{3}, \frac{2 + \boundb}{3}]}$$
\end{lemma}

\begin{proof}
 Given a multigraph $G$, for every edge $e = uv$ of $G$, we create two auxiliary vertices $w_e,z_e$, add three edges $uw_e$, $w_ez_e$, $z_e v$, and remove the edge $uv$, obtaining a resulting 
 graph $G'$. We claim that $G$ has a cut of size at least $k$ if and only if $G'$ has a cut of size at least $2m+k$. 
 
 Let $C$ be a cut of $G$ of size $k$. If $e = uv \not \in C$, we put
 $w_e$ and $z_e$ to the side of the cut opposite to $u$.
 This way two of the new edges, i.e., $uw_e$ and $z_ev$, are in the cut. 
 If $e = uv \in C$, then putting $w_e$, $z_e$ to the opposite sides of $u,v$,
 respectively, gives three edges in the cut. All together we obtain a cut of $G'$
 of size $2m+k$. 
 
 Let $C$ be a cut of $G'$ of size $2m+k$.
 Clearly, there are at least $k$ edges $uv$ of the original graph $G$,
 such that all the three edges $uw_e$, $w_e z_e$, and $z_e v$
 are in the cut $C$.
 Note that for such an edge $uv$ the vertex $w_e$ is on the side opposite to $u$,
 $z_e$ in on the side opposite to $w_e$ and $v$ is on side opposite to $z_e$,
 which means that $v$ in on the side of the cut opposite to $u$.
 Consequently a restriction of $C$ to vertices of $G$ gives a cut of size $k$.
\end{proof}


Lemmas~\ref{lem:3sat-nae4sat},~\ref{lem:nae4sat-nae3sat},~\ref{lem:nae3sat-mmaxcut}, and~\ref{lem:mmaxcut-maxcut} together prove Theorem~\ref{thm:hard-gapmaxcut}.
\end{proof}

Thus, we infer:

\begin{theorem} \label{thm:eth-maxcut} 
Unless ETH fails, there exist $c \in \N$, and $0 \leq \bounda < \boundb \leq 1$ such that there is no $2^{\mathcal{O}(\frac{m}{\log^c(m)})}$-time algorithm for \gapmaxcut$_{[\bounda,\boundb]}$.
\end{theorem}


Finally, we modify the reduction by Garey et al. \cite{GareyJS76} from \maxcut to \olashort.
The construction from~\cite{GareyJS76} introduces a huge clique to the complement of the \maxcut instance; this clique is supposed to separate in the ordering the sides of an optimum max-cut solution. 
The clique has to be large enough so that the cost of any edge going over the clique eclipses the ``noise'' contribution that is given by internal ordering of parts on the left and on the right. 
For this reason the clique is chosen to be of size $\Theta(n^4)$, which imposes
a significant blow-up in the instance size and immediately prevents from obtaining a desired subexponential hardness result.

This is precisely the moment where we exploit that we are working with an instance
of \maxcut with a gap, as starting with a gap instance,
we can accommodate the noise in the gap, and therefore we only need the clique to be of linear size.
Note that in the following theorem the number of vertices produced is linear,
however the produced instance might be dense, hence we cannot use the notation $\linleq$
and formulate the reduction properties explicitly.

\begin{theorem} \label{thm:red-maxcut-ola}
There is a polynomial time algorithm,
which given an instance $I$ of $\gapmaxcut_{[\bounda, \boundb]}$ with $n$ vertices
produces an instance $I'$ of $\olashort$ with $\mathcal{O}(n)$ vertices,
such that if $I$ admits a cut with at least $\beta m$ edges, then $I'$ is a yes-instance,
and if $I$ does not have a cut with at least $\alpha m$ edges, then $I'$ is a no-instance.
\end{theorem}

\begin{proof} 
 Given the input graph $G$, we create a corresponding instance of \olashort, $G'$ as follows. We take a complement of $G$, i.e., $\overline G$ and add a disjoint clique $C$ 
 of size $Mn$, for $M = \lceil \frac{2}{\boundb - \bounda} \rceil$ and fully connect it to $\overline{G}$.
  
 We will show that $G$ has a cut of size at least $\boundb m$ if and only if $G'$ has a linear arrangement of cost at most ${(M+1)n+1 \choose 3} - \boundb m \cdot Mn$,
 that will prove the equivalence of the instances.

 Let $(A,B)$ be a cut of size at least $\boundb m$. We define a linear arrangement $\aorder$ as follows. First, we list the vertices of $A$, then we list the vertices of the clique $C$
 and then the vertices of $B$. 
 The order of vertices inside $A$, $B$ and $C$ can be arbitrary. Then by counting the costs of all the edges in the clique on $(M+1)n$ vertices
 we obtain:
 
 $$\sum_{uv \in E(G')} |\aorder(u) - \aorder(v)| + \sum_{uv \in E(G)} |\aorder(u) - \aorder(v)| = {(M+1)n+1 \choose 3} (*)$$ 
 
 Now, the costs of each edge going across the clique, i.e., each edge in $E_{G'}(A,B)$, is at least $Mn$, thus $\sum_{uv \in E(G)} |\aorder(u) - \aorder(v)| \geq \boundb m \cdot Mn$, and it implies the desired inequality.
 
 In the other direction, let $\aorder$ be a linear arrangement of $G'$ of cost at most ${(M+1)n+1 \choose 3} - \boundb m \cdot Mn$. 
 First, we will prove that there exists an optimum linear arrangement $\aorder$, such that the vertices of $C$ has to be in the same consecutive block of $\aorder$. 
 Before that, we will introduce notation regarding to linear arrangements, also used in other proofs.
 
We assume the vertices of $G'$ are ordered from left to right according to a linear arrangement $\aorder$. 
When we speak about the $i$-th vertex (from the left), we mean the vertex mapped to the number $i$ by $\aorder$. 
A set of vertices $U$ is \term{consecutive in $\aorder$}, if $\aorder(U) = \{p, p + 1, \ldots, q - 1, q\}$ for $p, q \in \N$. 
The set of all vertices that are to the left of every vertex from some set $U$ is called \term{vertices to the left of $U$} and denoted by $L(U)$.
Similarly, we define \term{vertices to the right of $U$} and denote them by $R(U)$. 
\term{A block of $U$} is any inclusion-wise maximal non-empty subset of $U$ that is consecutive in $\aorder$. 
\term{The left-most block of $U$} is the block of $U$ whose vertices are mapped to the smallest values by $\aorder$. 
\term{Second left-most block of $U$} is the first block of $U$ to the right of the left-most block of $U$. 
\term{Inner block of $U$} is the set of all vertices from $\vertices[G'] \setminus U$ located simultaneously to the right the left-most block of $U$ and to the left of the second left-most block of $U$ (in the case when $U$ forms a single block, the inner block does not exist).

 \begin{claim}
\label{c.clique_together} 
 There exists an optimum linear arrangement $\aorder$ of $G'$ such that the vertices of $C$ are consecutive in $\aorder$.  
 \end{claim}
 
\begin{proof} 

   Let us choose $\aorder$ to be an optimum linear arrangement of $G'$ that minimize the number of vertices of $G$ that lie between the vertices of $C$,
   i.e. $f(\aorder) = |\{v \in V(G)  :  \aorder(u) < \aorder(v) < \aorder(w),  u,w \in V(C)\}|$ is the smallest possible.
   We claim that the vertices of $C$ are consecutive in $\aorder$, i.e. $f(\aorder) = 0$.
  
  Assume that this is not the case. Let $X$ be the inner block of $C$ and we consider the following two cases: 
  \begin{itemize}
   \item $|E_{G}(L(X), X)| \leq |E_{G}(X, R(X))|$
   \item $|E_{G}(L(X), X)| > |E_{G}(X, R(X))|$
  \end{itemize}

  In each of them, we will create another linear arrangement that will contradict the assumptions we made.

\begin{figure}[h]
\centering

\begin{tikzpicture}[scale=0.4,auto,node distance=2.8cm,semithick]

\path (0.5, 0) .. controls (0.5, 4) and (20.5, 4) .. (20.5, 0) -- (15.5, 0) .. controls (15.5, 0.7) and (14.5, 0.7) .. (14.5, 0) [fill=darkgraybluecolor,opacity=0.5];
\path (15.5, 0) .. controls (15.5, 4) and (28.5, 4) .. (28.5, 0) -- (21.5, 0) .. controls (21.5, 0.7) and (20.5, 0.7) .. (20.5, 0) [fill=darkgraybluecolor,opacity=0.5];



\draw[decorate, decoration={brace}, yshift=1.825cm] (0.75,2.8) -- node[above=0.4ex] {$G'$} (28.25,2.8);
\draw[decorate, decoration={brace,mirror}, yshift=1.825cm] (0.75,-4.2) -- node[below=0.4ex] {$L(C)$} (7.25,-4.2);
\draw[decorate, decoration={brace,mirror}, yshift=1.825cm] (8.75,-4.2) -- node[below=0.4ex] {$C'$} (14.25,-4.2);
\draw[decorate, decoration={brace,mirror}, yshift=1.825cm] (15.75,-4.2) -- node[below=0.4ex] {$X$} (20.25,-4.2);
\draw[decorate, decoration={brace,mirror}, yshift=1.825cm] (21.75,-4.2) -- node[below=0.4ex] {$R(X)$} (28.25,-4.2);
\path (7.9,2.75) node[above=0] {$E_{G'}(L(X), X)$};
\path (21.9,2.75) node[above=0] {$E_{G'}(X, R(X))$};
\path (22,-2.25) node[above=0] {$u$};

\path (4,0) ellipse (3.5 and 1.5) [draw=black] [fill=white];
\path (11.5,0) ellipse (3 and 1.5) [draw=black] [fill=white];
\path (18,0) ellipse (2.5 and 1.5) [draw=black] [fill=white];
\path (25,0) ellipse (3.5 and 1.5) [draw=black] [fill=white];

\foreach \s in {1, ..., 28}
{
	\ifnum \s=8 \else 
	\ifnum \s=15 \else 
	\ifnum \s=21 \else
		\node[vertex] (v\s) at (\s, 0) {}; 
	\fi \fi \fi
}

\end{tikzpicture}

\caption[]{The situation in the proof of Claim \ref{c.clique_together}.}
\label{f.clique_together}
\end{figure}
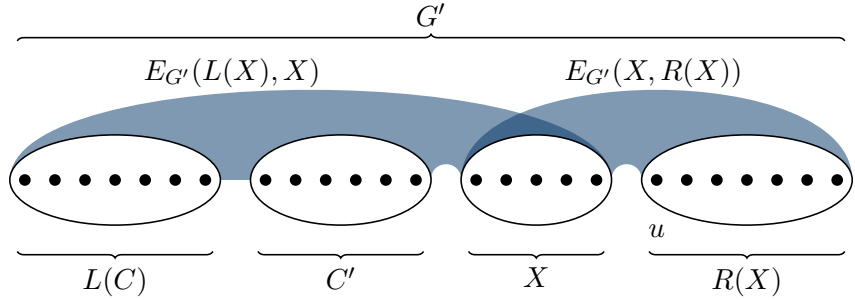
  
  If the first case occurs, let $C'$ be the left-most block of $C$. We swap places of $C'$ and $X$. Then, 
  we consider how it affects the cost of $\pi$. First, we note that by $(*)$, the optimality of $\pi$ can be equivalently restated as maximizing the sum of cost of edges of $G$, 
  i.e. $\pi$ is optimal if $\sum_{uv \in E(G)} |\aorder(u) - \aorder(v)|$ is the largest possible.
  We claim that the value of this sum has not decreased.
  Clearly, the cost of edges with both endpoints in $V(G) \setminus X$ or $X$ does not change. 
  Hence, it is enough to inspect the edges with exactly one endpoint in $X$. 
  Every such edge goes from $X$ to $L(X)$ or from $X$ to $R(X)$. An edge going from $X$ to $L(X)$ has decreased its length by $|C|$, and 
  an edge going from $X$ to $R(X)$ has increased its length by $|C|$. Thus, as $|E_{G}(L(X), X)| \leq |E_{G}(X, R(X))|$, the overall contribution 
  of those edges has not decreased, but the number of vertices between $C$ has dropped by $|X|$, leading to the contradiction.
  
  If the second case occurs, let $C'$ be the second-left most block of $C$. Similarly we swap places of $C'$ and $X$. Again, only affected edges
  are the ones that have exactly one endpoint in $X$. The inequality $|E_{G}(L(X), X)| > |E_{G}(X, R(X))|$ imples that $\sum_{uv \in E(G)} |\aorder(u) - \aorder(v)|$
  has increased at least by $|X|$, contradicting the optimality of $\aorder$.
  
  Thus, we have proved that $C$ has to be consecutive in $\aorder$.
  
\end{proof}

Let $\aorder$ be the optimum linear arrangement from the claim above. We define a cut $(A,B)$, by taking $A$ to be $L(C)$ and $B$ to be $R(C)$. Then by $(*)$ we know:
$$\boundb m \cdot Mn \leq \sum_{uv \in E(G)} |\aorder(u) - \aorder(v)| = $$ 
$$\sum_{uv \in E(G), \atop u,v \in A} |\aorder(u) - \aorder(v)| + \sum_{uv \in E(G), \atop u,v \in B} |\aorder(u) - \aorder(v)| + \sum_{uv \in E(G), \atop u \in A, v \in B} |\aorder(u) - \aorder(v)| <$$
$$nm + (Mn + n) \cdot |E_G(A,B)| \leq 2nm + Mn \cdot |E_G(A,B)|$$
Thus:
$$|E_G(A,B)| > \boundb m - \frac{2m}{M} \ge \boundb m - (\boundb - \bounda) \cdot m = \bounda m $$
By the promise given by the gap problem we know that if $G$ admits a cut greater than $\alpha m$,
then it actually admits a cut of size at least $\beta m$. 
\end{proof}


The proof of Theorem~\ref{thm:eth-ola} follows immediately from Theorems~\ref{thm:eth-maxcut}
and~\ref{thm:red-maxcut-ola}.

%% file: sec_sparse.tex
\def\minbola{T}

\section{Sparse reduction}\label{sec:sparse} 

We now introduce a polynomial-time Turing reduction from a gap version of the \minb\ problem on $d$-regular graphs to \ola. 
Its key property is that the instances of the former problem result in instances of \ola\ with linear number of vertices and bounded degree.
Even though the created instance of \ola is a multigraph, it does not cause any additional difficulties in further reductions described in Section~\ref{sec:completions}.
The reduction allows one to distinguish between instances of \minb\ with at most $\bounda m$ edges and at least $\boundb m$ edges in the optimum cut for some fixed choice of $0 \leq \bounda < \boundb \leq 1$ by solving the resulting instance of \ola. 
This relates \ola\ to the following hypothesis:  
\begin{rephypothesis}{h.bisection} 
There exist $0 \leq \bounda < \boundb \leq 1$, and an integer $d > \frac{4}{\boundb-\bounda}$, such that there is no $2^{o(n)}$-time algorithm for \gapminbd[d]$_{[\alpha, \beta]}$.
\end{rephypothesis} 

The main result of this section is: 
\begin{reptheorem}{thm:hyp-ola}
Unless Hypothesis~\ref{h.bisection} fails, there exists an integer $d \in \N$ such that there is no $2^{o(n)}$-time algorithm solving \oladshort[d] in multigraphs.
\end{reptheorem}

 

We first describe a transformation $\minbola(\cdot)$ from an instance $G$ of \minb\ to an instance of \ola\ that forms the key component of our reduction. 
Then, we introduce several technical claims about its properties. 
Finally, we prove Theorem \ref{thm:main} by showing how to decide the instances of \gapminb\ based on the cost of the optimum arrangement of $\minbola(G)$. 
The fact that our reduction exhibits only a linear increase in the size of the instance is crucial in achieving the $2^{\Omega(n)}$ bound. 

The result of transformation $\minbola(\cdot)$ is influenced by several parameters. 
The choice of their values is deferred to the proof of Theorem
\ref{thm:main}. 
Consider an instance $G$ of the \minb\ problem, where $G$ is a $d_G$-regular graph. 
Assume $\vertices[G] = \{ v_1, \ldots, v_n \}$. 
The transformation produces a graph $G' := \minbola(G)$ with 
the vertex set $\{v_1, \ldots, v_n, x_1, \ldots, x_{Z \cdot \phin}\}$, where $Z \in N$ and $\varphi \in (0,1)$ are constants chosen later. 
\begin{figure}[h]
\centering

\begin{tikzpicture}[scale=0.4,auto,node distance=2.8cm,thick]

\foreach \s in {1, ..., 6}
{
	\node[vertex] (v\s) at (\s - 0.5, 0) {}; 
}

\foreach \s in {7, ..., 24, 26, 27, 28, 29, 30}
{
	\node[vertex] (v\s) at (\s, 0) {}; 
}

\fill[dot] (24.75, 0) circle (1.4pt);
\fill[dot] (25, 0) circle (1.4pt);
\fill[dot] (25.25, 0) circle (1.4pt);

\fill[dot] (6, 0) circle (1.4pt);
\fill[dot] (6.25, 0) circle (1.4pt);
\fill[dot] (6.5, 0) circle (1.4pt);

\path[greenedge] (v1) -- (v2);
\path[greenedge] (v5) -- (v6); 
\path[greenedge] (v7) -- (v8); 
\path[greenedge] (v2) to[out=-45,in=-135] (v7); 
\path[greenedge] (v1) to[out=30,in=150] (v6); 
\path[greenedge] (v3) to[out=-45,in=-135] (v7); 
\path[greenedge] (v4) to[out=-45,in=-135] (v8); 
\path[greenedge] (v5) to[out=-45,in=-135] (v9); 
\path[greenedge] (v3) to[out=30,in=150] (v9); 
\path[greenedge] (v4) to[out=30,in=150] (v9); 

\draw (4.75,0) ellipse (4.66 and 1.525);
\draw (12,0) ellipse (2.45 and 1.3);
\draw (17,0) ellipse (2.45 and 1.3);
\draw (22,0) ellipse (2.45 and 1.3);
\draw (28,0) ellipse (2.45 and 1.3);

\draw (5, -2.5) node {$G$};
\draw (12, -2.5) node {$H_1$};
\draw (17, -2.5) node {$H_2$};
\draw (22, -2.5) node {$H_3$};
\draw (28, -2.5) node {$H_Z$};

\draw[decorate, decoration={brace,mirror}, yshift=-3.2cm] (9.85,0) -- node[below=0.4ex] {$H \approx G_{Z\lceil \varphi n \rceil, d_H}$} (30.15,0);
\draw[decorate, decoration={brace,mirror}, yshift=-4.85cm] (0.25,0) -- node[below=0.4ex] {$G'$} (30.15,0);
\draw[decorate, decoration={brace}, yshift=1.825cm] (19.85,0) -- node[above=0.4ex] {$H_i \approx G_{\lceil \varphi n \rceil, d_{H_i}}$} (24.15,0);

\draw[blueedge] (v1) .. controls (2, 4.5) and (8, 4.5) .. (v10); 
\draw[blueedge] (v2) .. controls (3, 4.5) and (8, 3.5) .. (v10); 
\draw[blueedge] (v3) .. controls (4, 4.5) and (9, 4.5) .. (v11); 
\draw[blueedge] (v4) .. controls (5, 4.5) and (9, 3.5) .. (v11); 
\draw[blueedge] (v5) .. controls (6, 4.5) and (10, 4.5) .. (v12); 
\draw[blueedge] (v6) .. controls (7, 4.5) and (10, 3.5) .. (v12); 
\draw[blueedge] (v7) .. controls (8, 4.5) and (11, 4.5) .. (v13); 
\draw[blueedge] (v8) .. controls (9, 4.5) and (11, 3.5) .. (v13); 
\draw[blueedge] (v9) .. controls (10, 4.5) and (12, 4.5) .. (v14); 

\end{tikzpicture}

\caption[]{
The resulting instance $G' = T(G)$ after applying the reduction. 
The original graph $G$ is an induced subgraph of $G'$ with its edges shown in green. 
The edges of the bipartite graph added between $V(G)$ and $V(H_1)$ are shown in blue. 
}
\label{f.sparse}
\end{figure}
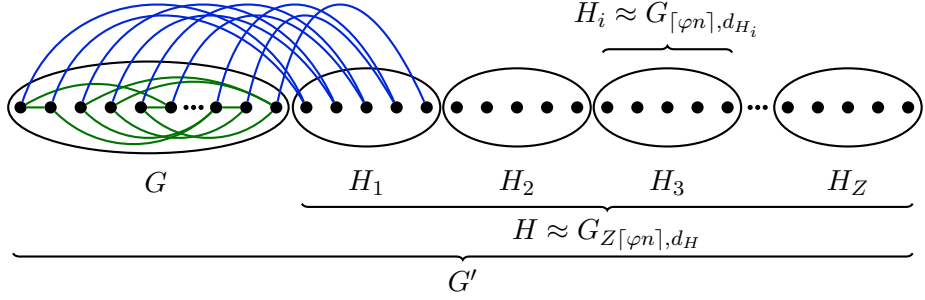
Note that $G'$ contains the vertices of $G$. 
Indeed, we are going to construct the edge-set in such a way that $G$ is an induced subgraph of $G'$. 
It is actually convenient to introduce notation for some of the induced subgraphs of $G'$. 
The subgraph with the vertex set $\{x_1, \ldots, x_{Z \cdot \phin}\}$ is denoted by $H$. 
The graph $H$ is (arbitrarily) divided into $Z$ disjoint induced subgraphs $H_i$ of size $\phin$ each, for some constants $Z \in \N$ and $\varphi \in (0,1)$. 

The result of the transformation is illustrated in Figure~\ref{f.sparse}. 
The edge-set of $G'$ is constructed as follows: 

\begin{itemize}
\item The induced subgraph of $G'$ on $\{v_1, \ldots, v_n\}$ is $G$. 
\item We construct a $\dH$-regular expander $G_{|H|,\dH}$ using Theorem \ref{thm:expander-factory} 
satisfying $h(G_{|H|,\dH}) \ge p_H$, and add its edges on the vertices of $H$ (the value $p_H$
will be determined later). 
\item For each $i \in \{1, \ldots, Z\}$ we construct a $\dHi$-regular expander $G_{|H_i|,\dHi}$ using Theorem~\ref{thm:expander-factory} satisfying $h(G_{|H_i|,\dHi}) \ge p_{H_i}$, and add its edges on the vertices of $H_i$ (the value $p_{H_i}$ will be determined later).
\item For each $i \in \{1, \ldots, Z\}$ we add a bipartite graph on parts $\vertices[G]$ and $\vertices[H_i]$ such that all vertices of $\vertices[G]$ have degree 1 in this bipartite graph and the degrees of vertices 
from $\vertices[H_i]$ differ by at most 1. 
We denote the maximum degree of the $\vertices[H_i]$ part of this added bipartite graph by $\dHG$. 
It is at most $\lceil \frac{1}{\varphi} \rceil$.  
\end{itemize}

Note that we are constructing a multigraph, that is when an edge is to be added several
times in the construction process, we keep all its copies.

In the following part, we give the proof of correctness of the transformation, as
well as determine the parameters driving the reduction. In the proof, we first show that 
the vertices of $H$ have to be consecutive in an optimum linear arrangement.
Next, we show that actually vertices of each small expander $H_i$ are consecutive 
in such an ordering.
This is crucial when we analyze the change of the cost of the ordering
when moving a vertex of $G$ from one side of $H$ to the other,
and in turn prove that in an optimum ordering the parts of $G$ to the left
and to the right of $H$ are almost of the same size.
Interestingly, in our reduction we have to use the hypothetical oracle solving the decision version
of \ola to find the cost of an optimum ordering of an auxiliary graph by using binary search.

The constructed $G'$ is influenced (apart from the input graph $G$) by our choice of parameters $Z, \varphi, p_H$ and $p_{H_i}$,
which in turn influence $\dH$ and $\dHi$ by Theorem~\ref{thm:expander-factory}.
The lemmas below impose a particular structure on the optimum linear arrangement of $G'$, provided certain inequalities between these parameters are satisfied. 
Eventually, the lemmas are employed in the proof of the main theorem of this section.

The following technical \textit{Swapping Lemma} establishes a condition on degrees in two consecutive sets of $G'$ under which the swapping of the two sets results in a decreased cost of the ordering. 
\begin{lemma}[Swapping Lemma]
\label{l.swapping}
Consider an ordering $\aorder$\ of any finite graph $G'$. 
Assume that the sets $X, Y \ss \vertices[G']$ are consecutive and $X$ immediately precedes $Y$. 
Let $L := L(X)$ and $R := R(Y)$. 
Assume 
\begin{itemize}
\item 
the value $P_X$ upper bounds the degree of vertices from $X$ in the induced bipartite subgraph $G'[L, X]$, 
\item 
$P_C$ is an upper bound on the maximum degree of $G'[X, Y]$, and  
\item 
$P_Y$ is an upper bound on the degree of a vertex from $Y$ in $G'[Y, R]$. 
\end{itemize}
Finally, let $p$ be a lower bound on the average degree of a vertex from $X$ in $G'[X, R]$. 
Then the inequality
$p > P_X + 2P_C + P_Y $
implies that swapping the vertices of $X$ with the vertices of $Y$ in the order specified by $\aorder$\ results in a decrease in the cost of the ordering. 
\end{lemma}
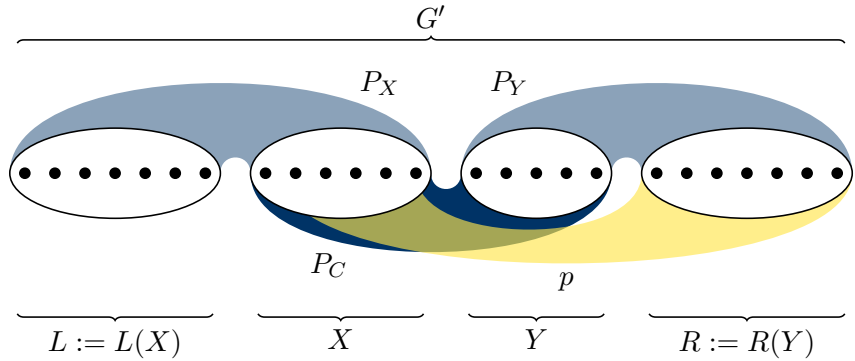
\begin{figure}[h]
\centering

\begin{tikzpicture}[scale=0.4,auto,node distance=2.8cm,semithick]

\path (0.5, 0) .. controls (0.5, 4) and (14.5, 4) .. (14.5, 0) -- (8.5, 0) .. controls (8.5, 0.7) and (7.5, 0.7) .. (7.5, 0) [fill=graybluecolor];
\path (15.5, 0) .. controls (15.5, 4) and (28.5, 4) .. (28.5, 0) -- (21.5, 0) .. controls (21.5, 0.7) and (20.5, 0.7) .. (20.5, 0) [fill=graybluecolor];

\path (8.5, 0) .. controls (8.5, -3.5) and (20.5, -3.5) .. (20.5, 0) -- (15.5, 0) .. controls (15.5, -0.7) and (14.5, -0.7) .. (14.5, 0) [fill=deepbluecolor];

\path (9.5, 0) .. controls (9.5, -4) and (28.5, -4) .. (28.5, 0) -- (21.5, 0) .. controls (21.5, -2.5) and (14, -2.5) .. (14, 0) [fill=yellowishcolor,opacity=0.65];

\draw[decorate, decoration={brace}, yshift=1.825cm] (0.75,2.5) -- node[above=0.4ex] {$G'$} (28.25,2.5);
\draw[decorate, decoration={brace,mirror}, yshift=1.825cm] (0.75,-6.5) -- node[below=0.4ex] {$L := L(X)$} (7.25,-6.5);
\draw[decorate, decoration={brace,mirror}, yshift=1.825cm] (8.75,-6.5) -- node[below=0.4ex] {$X$} (14.25,-6.5);
\draw[decorate, decoration={brace,mirror}, yshift=1.825cm] (15.75,-6.5) -- node[below=0.4ex] {$Y$} (20.25,-6.5);
\draw[decorate, decoration={brace,mirror}, yshift=1.825cm] (21.75,-6.5) -- node[below=0.4ex] {$R := R(Y)$} (28.25,-6.5);
\path (17.1,2.25) node[above=0] {$P_Y$};
\path (12.8,2.25) node[above=0] {$P_X$};
\path (11.1,-3.8) node[above=0] {$P_C$};
\path (19,-4.25) node[above=0] {$p$};

\path (4,0) ellipse (3.5 and 1.5) [draw=black] [fill=white];
\path (11.5,0) ellipse (3 and 1.5) [draw=black] [fill=white];
\path (18,0) ellipse (2.5 and 1.5) [draw=black] [fill=white];
\path (25,0) ellipse (3.5 and 1.5) [draw=black] [fill=white];

\foreach \s in {1, ..., 28}
{
	\ifnum \s=8 \else 
	\ifnum \s=15 \else 
	\ifnum \s=21 \else
		\node[vertex] (v\s) at (\s, 0) {}; 
	\fi \fi \fi
}

\end{tikzpicture}

\caption[]{
The vertex set of $G'$ is partitioned into four sets, $L, X, Y,$ and $R$, in Lemma \ref{l.swapping}.
The bounds $P_X, P_Y, P_C$ (upper bounds), and $p$ (a lower bound) on the degrees of the induced bipartite subgraphs are also shown. 
}
\label{f.swapping}
\end{figure}
\begin{proof}
The situation is illustrated in Figure \ref{f.swapping}. 
The length of all edges connecting a pair of vertices from one of the sets $L, X, Y, R$ remains unchanged after swapping $X$ and $Y$ in the ordering. 
The same holds for edges connecting $L$ with $R$. 
The length of each edge connecting $X$ and $Y$ increases by at most $|X| + |Y| \le 2\max\{|X|, |Y|\}$.
The cost of each edge connecting $X$ and $L$ increases by at most $|Y|$. 
Similarly, the cost of each edge connecting $Y$ and $R$ increases by at most $|X|$. 
On the other hand, the edges connecting $X$ and $R$ are shortened, each by $|Y|$. 
The upper bounds on maximum degrees and the lower bound on average degree from the statement of the lemma allow us to lower bound the decrease in the cost of the ordering after the swap is performed. 
For example, the decrease in total cost of the edges connecting $X$ with $R$ is at least $p|X||Y|$.
The decrease in cost after swapping is at least
$$p|X||Y| - 2 \min\big\{|X|, |Y|\big\} P_C \max\big\{|X|, |Y|\big\} - |X| P_X |Y| - |Y| P_Y |X|,$$
which is equal to
$$|X||Y|(p - 2 P_C - P_X - P_Y).$$
Assuming the inequality from the lemma, this is strictly larger than zero. 
\end{proof}

We now make several claims about the optimum ordering $\aorder$ of $G' := \minbola(G)$, where $G$ is a $\dG$-regular graph. 
Recall that that $Z, \varphi, p_H, p_{H_i}, \dH,$ and $\dHi$ are the parameters of the transformation $\minbola(\cdot)$
still to  be determined. 

\begin{lemma}
\label{l.htogether} 
If $p_H > 3\dHG + 3Z + \dG$ and $\aorder$\ is an optimum linear arrangement of $G'$, then $V(H)$ is consecutive in $\aorder$. 
\end{lemma}
\begin{proof}
Suppose $\vertices[H]$ is not consecutive in $\aorder$. 
Consider the left-most block of $\vertices[H]$ and denote its elements by $X$. 
We can assume that $|X| \le \half{|H|}$ -- otherwise we take the right-most block of $\vertices[H]$ and proceed with a mirrored version of the following argument. 
Denote by $Y$ the inner block of $\vertices[H]$ 
and set $L := L(X), R := R(Y)$. 
The following choice of values satisfies the assumptions on degree upper-bounds of the Lemma \ref{l.swapping}: 
\begin{align*}
P_X & := \dHG, & 
P_Y & := \dG + Z, & 
P_C & := \dHG + Z.
\end{align*}
Since $H$ is an expander, $|X|\leq |H|/2$ and $H \setminus X \ss R$, we take $p = p_H$. 
It remains to show the inequality from the statement of the Swapping Lemma.
We have: 
$$p = p_H > 3\dHG + 3Z + \dG = P_X + 2P_C + P_Y.$$
Thus, we can swap $X$ and $Y$ and decrease the cost of the ordering. 
This contradicts the optimality of $\aorder$.
\end{proof} 

\begin{lemma} 
\label{l.hitogether} 
Let $i \in \{1, \ldots, Z\}$.
If $p_{H_i} > \dHip + 4\dH + 2\dHG$, $\aorder$\ is an optimum linear arrangement of $G'$, and for each $i' < i$ the vertices of $H_{i'}$ are consecutive in $\aorder$, then the vertices of $H_i$ are consecutive in~$\aorder$. 
(For the purposes of this lemma, we set $d_{H_{Z+1}} := 0$.) 
\end{lemma} 

\begin{proof}
Assume $H_i$ not to be consecutive and denote by $X$ the left-most block of $\vertices[H_i]$ in $\aorder$. 
Similarly to the situation in the proof of Lemma~\ref{l.htogether}, 
we can assume $|X|\leq |H_i|/2$. 
(Otherwise we consider the right-most block instead and proceed with a mirrored version of the argument.) 
We show that moving $X$ to the right decreases the cost of the arrangement. 

Denote by $u$ the vertex positioned by $\aorder$\ immediately to the right of $X$. 
Due to Lemma \ref{l.htogether}, we know $u \not\in \vertices[G]$. 
Therefore, $u \in \vertices[H_j]$ for $j \neq i$. 
We distinguish two cases: either $j < i$ or $j > i$. 

Suppose that $u \in \vertices[H_j]$ for $j < i$. 
Note that $H_j$ is consecutive in $\aorder$. 
We set $Y := \vertices[H_j], L := L(X),$ and $R := R(Y)$.  
Again, we employ the Swapping Lemma. 
The following degree upper-bounds satisfy its assumptions: 
\begin{align*}
P_X & := \dHG + \dH, &
P_Y & := \dHG + \dH, &
P_C & := \dH.
\end{align*}
Since $H_i$ is an expander, $|X|\leq |H_i|/2$ and $\vertices[H_i] \setminus X \ss R$, we can set the average degree lower-bound $p$ to $p_{H_i}$. 
By the inequality from the statement of this lemma, we have
$$p = p_{H_i} > 4\dH + 2\dHG = P_X + 2P_C + P_Y.$$
Thus, the inequality from the Swapping Lemma holds and we can use it to decrease the cost of ordering, contradicting the optimality of $\aorder$. 

Suppose therefore that $u \in \vertices[H_j]$ for $j > i$. 
We now use Lemma \ref{l.swapping} again to move the block $X$ one position to the right, effectively swapping $X$ and $Y := \{u\}$. 
We set $L := L(X), R := R(Y)$.
This time, we set: 
\begin{align*}
P_X & := \dHG + \dH, &
P_Y & := \dHG + \dH + \dHip, & 
P_C & := \dH.
\end{align*}
Similarly to the previous cases, we set $p := p_{H_i}$. 
The inequality from the Swapping Lemma is again satisfied: 
$$p = p_{H_i} > 4\dH + 2\dHG + \dHip = P_X + 2P_C + P_Y.$$
Once more, we get a contradiction with the optimality of $\aorder$. 

\end{proof}

Due to Lemma \ref{l.htogether} we know that an optimum linear arrangement of $G'$ places vertices of $H$ consecutively, assuming the inequalities from its statement are satisfied. 
Furthermore, by iterating Lemma \ref{l.hitogether} we get that within $H$, the vertices of each $H_i$ are grouped together in the arrangement. 
To be precise, in the optimum ordering the subgraphs $H_i$ are placed in the order $H_{\ell_1}, H_{\ell_2}, \dots , H_{\ell_Z}$, where $(\ell_1,\ell_2,\dots, \ell_Z)$ is a permutation of $\{1,2,\dots, Z\}$.
The vertices of $G$ can thus be only to the left of $H$ or to its right. 
The next lemma shows that $H$ divides the graph $G$ into two roughly equal parts. 

\begin{lemma}
\label{l.imbal} 
Assume $G'$ has been constructed by the transformation $\minbola(\cdot)$ with parameters satisfying the inequalities from the statements of Lemmas \ref{l.htogether} and \ref{l.hitogether}
as well as $Z \varphi \ge 2$.
Moreover, assume $\gamma  = 3 \varphi \dG$. 
Consider an optimum linear arrangement $\aorder$\ of $G'$
and set $A := L(H), B := R(H)$. 
Then $\big| |A| - |B| \big| \le \gamma n$.
\end{lemma}

\begin{proof}
Assume the imbalance $\big| |A| - |B| \big|$ is strictly bigger than $\gamma n$. 
Without loss of generality, assume $|A| > |B|$. 
We consider the vertex $u$ such that $\aorder(u) = 1$ (i.e., the one placed on the left side of the arrangement). 

Moving $u$ to the right-most position results in the following changes in the cost of the arrangement. 
The cost associated with the edges of $G$ might be increased by at most $\dG (Z \phin + n) \leq \dG (Z \varphi n + Z + n)$.
Note that for sufficiently large $n$ we have $\frac{Z \varphi n}{2} \ge Z$ and by the assumption of the Lemma
we have $\frac{Z \varphi n}{2} \ge n$, therefore the cost assosiated with the edges of $G$ increases by at most $\dG (Z \varphi n + Z + n) \le 2 \dG Z \varphi n$. 

In addition to this, the vertex $u$ is connected to precisely one vertex $v_i$ of each $H_{\ell_i}$. 
Before moving $u$, the edge cost of $uv_i$ was $|A|+ (i-1) \phin + r_i$,
for some $0 \leq r_i \leq \phin-1$ and after, it becomes $|B|+((Z - i + 1) \phin - r_i)$, thus the contribution of all those edges has changed by:
$$(|B| - |A|) Z + \sum_{i=1}^{Z} ((Z - i + 1) \phin - r_i) - \sum_{i=1}^{Z} ((i-1) \phin + r_i)=(|B| - |A|) Z + \sum_{i=1}^{Z} (\phin - 2r_i) < $$

$$ - \gamma n Z +Z \phin \stackrel{\gamma = 3 \dG \varphi}{=} -3\dG Z \varphi n + Z \phin  \stackrel{Z \phin < \dG Z \varphi n}{<} - 2\dG Z \varphi n$$
Note that in the last inequality we have used the assumption $\phin < \dG \varphi n$ which holds for $\dG \ge 2$ and suffiently large $n$, as $\varphi$ is a constant.

Combining this cost change with the cost increase of edges of $G$ being at most $2 \dG Z \varphi n$, we obtain that moving $u$ to the right-most position causes 
the decrease in the cost of $\aorder$, leading to a contradiction. 
\end{proof}


Being equipped with all the required tools, we are ready to prove Theorem~\ref{thm:hyp-ola}.

%% file: sec_sparse_theorem.tex
\begin{proof}[Proof of Theorem \ref{thm:hyp-ola}]
We prove the theorem by introducing a Turing reduction from \gapminbd[d]$_{[\bounda, \boundb]}$ to \oladshort[d'], where \begin{align} \label{eq:dg} d > \frac{4}{\boundb - \bounda} \end{align} and $d'$ is some constant 
depending on $d$, $\bounda$, and $\boundb$. The reduction proceeds in the following way. 

Given an instance $G$ of \gapminbd[d]$_{[\bounda, \boundb]}$, we create an equivalent instance $(G',k)$ of \oladshort[d'] based on the transformation $\minbola(\cdot)$ applied to $G$. 
The value $k$ will depend on, among other parameters, the cost of the optimal arrangement of the expander $H$ (recall this is an induced subgraph of $G'$). 
As mentioned in the beginning, the reduction is a Turing reduction. 
It is therefore equipped with an oracle solving the decision version of \olashort,
which in turn is enough to find the cost of an optimal ordering by using binary search (we will use it to compute $\mbox{OLA}(H)$).


We start by establishing the parameters of the transformation in the following order.
\begin{itemize}
  \item First, we set $\gamma := \frac{\boundb-\bounda}{4}$ and $\varphi := \frac{\gamma}{3 \dG}$.
  Note that $\gamma, \varphi \in (0,1)$ and $\gamma, \varphi$ satisfy the 
  condition $\gamma  = 3 \varphi \dG$ from Lemma~\ref{l.imbal}.
  \item Next, we set the value of $Z$ to the following integer:
    $$Z := \lceil \frac{2(2\bounda + 1)}{(\boundb - \bounda) \varphi}) \rceil\,.$$
    In particular $Z \geq \frac{2}{\varphi}$ satisfying the condition from Lemma~\ref{l.imbal}. Moreover, the following inequality holds:
    \begin{align}
    \label{eq:z} 2(2\bounda + 1) \leq (\boundb - \bounda) Z \varphi.
    \end{align}
  \item By construction $\dHG \leq \big\lceil \frac{1}{\varphi} \big\rceil = \big\lceil \frac{3 \dG}{\gamma} \big\rceil$. 
  \item Next, we define $p_H$, which in turn determines the value of $\dH$ by Theorem~\ref{thm:expander-factory}:
$$p_H   := 3\dHG + 3Z + \dG + 1\,.$$
The additive term $+1$ is just to make sure the inequality from Lemma~\ref{l.htogether} is strict and it could be replaced by any positive constant.
  \item Finally, we set the values $p_{H_i}$:
  $$p_{H_{i}}   := d_{H_{i+1}} + 4\dH + 2\dHG + 1,$$ in the decreasing order $i=Z,\ldots,1$, where for simplicity we denote $d_{H_{Z+1}} = 0$.  Note that the value of $p_{H_i}$ determines the value of $d_{H_i}$ via Theorem~\ref{thm:expander-factory}.
\end{itemize}

Observe that with the above definition of all the constants we satisfy all the 
requirements of Lemmas~\ref{l.htogether},~\ref{l.hitogether},~\ref{l.imbal}.

Let $G'$ be the result of the transformation $\minbola(G)$ with 
the above choice of parameters.
The value $k$ is then set as follows: 
\begin{align}
\label{eq:defk}
k := \mbox{OLA}(H) + \bounda m \cdot (Z \phin + n) + m \cdot \half{n} + \Big((\frac{n}{2}+1) \frac{n}{2} Z + n\sum_{i=1}^{Z} i \phin\Big)\,.
\end{align}
Note that to compute the value of $k$ we need use the oracle solving the decision
version of $\minbola$ in a binary search routine.
It remains to show that $G$ has a bisection of size at most $\bounda m$ if and only if 
$G'$ has an optimum linear arrangement of size at most $k$. 

Let us assume $G$ has a bisection with at most $\bounda m$ edges. 
We claim that $k$ is an upper bound on the cost of an optimum linear arrangement of $G'$. 
This is because it accounts for all costs associated with an ordering of $G'$ constructed from the optimum bisection of $G$. 
Denote by $A, B$ the partition of $V(G)$ corresponding to an optimum bisection of $G$. 
We create an ordering $\aorder$ that first lists all vertices of $A$, then the vertices of $H$ in the order of an optimum linear arrangement of $H$, and finally the vertices of $B$. 

The first term of~(\ref{eq:defk}) is the cost of all edges inside $H$. 
The second term upper bounds the cost of edges of $G$ between $A$ and $B$: 
there are at most $\bounda m$ of them and we upper bound the cost of every such edge by $|V(G')| = (Z \phin + n)$.
In the third term, we account for the cost of edges within $A$ and within $B$.
There are at most $m$ of them and since $|A| = |B| = \frac{n}{2}$ 
every such edge has cost at most $\half{n}$. 
The last term is an upper bound on the cost of edges connecting $G$ and $H$. 
Every vertex $v$ of $G$ has an edge to exactly one vertex of each $H_{\ell_i}$. 
If $v \in A$ we may bound its cost by $j(v) + i \phin$, where $j(v)$ is the length of the part 
of the ordering from $v$ to the first vertex in $H$. 
We first count the contribution of the $j(v)$-terms in the above expression for all choices of $v \in A$. 
Since $|A| = n/2$, summing over all $v \in A$ and $i=1,\ldots,Z$
we get $\sum_{j=1}^{|A|} j Z = \frac{(\frac{n}{2}+1) \frac{n}{2}}{2} Z$. 
The situation is analogous for $B$. 
The last term of~(\ref{eq:defk}) is obtained by summing 
the remaining edge costs $i \phin$ for all $i=1,\ldots,Z$ and $v \in G$. 
This proves the claimed upper bound.

In the other direction, we start with assuming that the graph $G'$ has an optimum arrangement $\aorder$ of cost at most $k$. 
The aim is to prove that $G$ has a bisection of size at most $\alpha m$.
Lemma \ref{l.htogether}, Lemma \ref{l.hitogether}, and Lemma \ref{l.imbal} together impose a certain structure on $\aorder$. 
Particularly, the vertices of $H$ are placed together in $\aorder$.
We use this fact to construct a bisection of $G$. 
Set $A := L(H)$ and $B := R(H)$, and if $|A| < |B|$, then replace $A$ with $B$. 
Note that $(A,B)$ is a partition of $V(G)$ and these two sets might have different sizes with the imbalance bounded by Lemma \ref{l.imbal}. 
We now bound the number of edges between $A$ and $B$. 
To this end we lower bound the cost of $\aorder$ in terms of $|E_{G'}(A,B)|$. 
Specifically, it is at least:
\begin{align}
\label{eq:lb}
\mbox{OLA}(H) + |E_{G'}(A,B)| \cdot Z \phin + \Big((\frac{n}{2}+1) \frac{n}{2} Z + n\sum_{i=1}^{Z} (i-1) \phin \Big)\,.
\end{align}
There, the first term of~(\ref{eq:lb}) accounts for the cost of the edges of $H$,
as it is lower bounded by $\mbox{OLA}(H)$. 
The second term of~(\ref{eq:lb}) is a lower bound on the cost of edges 
of $G$ going across the partition $(A, B)$. 
There are $|E_{G'}(A,B)|$ of such edges and each contributes at least $Z \phin $ to the cost. 
Recall that vertices of $H$ must have the following order $H_{\ell'_1}, H_{\ell'_2}, \dots, H_{\ell'_Z}$, where $(\ell'_1,\ell'_2, \dots, \ell'_Z)$ is some permutation of $\{1,2,\dots, Z\}$. 
The third term lower bounds the cost of edges connecting $G$ to $H$. 
Every vertex $v$ of $G$ has an edge to exactly one vertex from $H_{\ell'_i}$. 
Similarly to the analysis above, 
we lower bound the cost of such an edge by $j(v) + (i-1) \phin$. 
The contribution of $j(v)$'s from the above expression for all choices of $v \in G$ is 
equal to $\Big( \sum_{j=1}^{|A|} j + \sum_{j=1}^{|B|} j \Big) Z \geq \Big( 2 \sum_{j=1}^{\frac{n}{2}} j \Big) Z \geq (\frac{n}{2}+1) \frac{n}{2} Z$.
The remaining part is obtained by summing $(i-1) \phin$ over all $v \in G$ and $i=1,\ldots,Z$. 
Comparing~(\ref{eq:defk}) with~(\ref{eq:lb}) we obtain:
$$|E_{G'}(A,B)| \cdot Z \phin \leq \bounda m \cdot (Z \phin + n)  + m \cdot \half{n} + n\sum_{i=1}^{Z} i \phin
-n\sum_{i=1}^{Z} (i-1) \phin  =$$
$$\bounda  m \cdot Z \phin +  \frac{2(2 \bounda + 1) mn}{4}  + n\sum_{i=1}^{Z} \phin \stackrel{{\mathrm by}~(\ref{eq:z})}{\leq}
\bounda  m \cdot Z \phin +  \frac{(\boundb - \bounda) m \cdot Z \varphi n}{4}  + n \cdot Z \phin \stackrel{\varphi n \leq \phin}{\leq}$$
$$\bounda  m \cdot Z \phin +  \frac{(\boundb - \bounda) m}{4} \cdot Z \phin  + n \cdot  Z \phin \stackrel{n = \frac{2}{\dG}m < \frac{\boundb - \bounda}{2} m~{\mathrm by}~(\ref{eq:dg})}{<} 
\frac{\bounda + 3\boundb}{4}m \cdot Z \phin.$$

We infer that $|E_{G'}(A,B)| < \frac{\bounda + 3\boundb}{4} m$.
Next, we create a bisection $(A',B')$ of $G$ as follows. 
Let $C$ be a set of $\frac{|A|-|B|}{2}$ arbitrary vertices of $A$, then we put $A' = A - C$, $B' = B \cup C$. 
By moving these vertices we get at most
$\frac{\gamma  n}{2} \cdot \dG = \gamma m$ additional edges in the cut, thus:
$$|E_{G'}(A',B')| \leq |E_{G'}(A,B)| + \gamma m  < \frac{\bounda + 3\boundb}{4}m +  \frac{\boundb - \bounda}{4}m =  \boundb  m.  $$
Therefore, we have $|E_{G'}(A',B')| < \boundb m$. 
Since the problem considered is a gap problem, we know that $G$ admits either a bisection of size at most $\bounda m$ or there is no bisection of size less then $\boundb m$. 
Therefore, we conclude the instance admits a bisection of size at most $\bounda m$.
\end{proof}

%% file: sec_hardness.tex
 \section{Lower Bounds for Minimum Fill-in and Other Completion \protect\\ Problems}
 \label{sec:completions}

In this section we prove Theorems~\ref{thm:only-eth} and~\ref{thm:main},
that is prove conditional lower bounds (under ETH and under Hypothesis~\ref{h.bisection})
for parameterized completion problems such as
\minimumfillin, \chaincompletion, \properintervalcompletion, \intervalcompletion, \thresholdcompletion, \triviallyperfectcompletion.
As a starting point we use Theorems~\ref{thm:eth-ola} and~\ref{thm:hyp-ola},
hence our goal is to transform an instance of \olashort into an instance
of a graph completion problem.
The main reduction of this section, which transforms
an instance of \olashort to \chaincompletion 
is a slight modification of the reduction of Yannakakis~\cite{yannakakis},
with the only difference that on bounded degree instance of \olashort we obtain
linear number of vertices in the final instance of \chaincompletion. This fact is crucial to prove Theorem~\ref{thm:main} while for proof of Theorem~\ref{thm:only-eth} it is enough to use the original version of Yannakakis reduction.

\begin{definition}A bipartite graph $(A, B ,F)$ with vertices $A \uplus B$ and edges $F$ is a \emph{chain graph} if the set of vertices $A$ (called left side) can be ordered $v_1, v_2, \dots, v_n$ 
(called left order) in such a way that $N(v_1) \subseteq N(v_2) \subseteq \dots \subseteq N(v_n).$
\end{definition}

In the \chaincompletion problem given a bipartite graph $(A, B, F)$
one is asked to add a minimum number of edges $F' \subseteq A \times B \setminus F$ 
such that $(A, B, F\cup F')$ is a chain graph.

\begin{lemma}\label{lem:ola-to-chain}
There is a polynomial time algorithm,
which given an instance $I=(G=(V,E),k)$ of \ola
creates an equivalent instance $I'=(G'=(A, B, F), k')$ of 
\chaincompletion, such that the number of vertices
of $G'$ is bounded by $\Oh(\Delta_G \cdot |V|)$,
where $\Delta_G$ is the maximum degree of $G$.
The reduction works even if $G$ is a multigraph.
\end{lemma}

\begin{proof}
As the left side of $G'$ we take $A=V$.
For each vertex $v \in V$ create a set of $\Delta_G$ new vertices $S_v=\{ v_e : e \in \delta_G(v) \} \cup \{v_i : \deg(v) < i \le \Delta_G\}$.
We define $B$ as the union of all the sets $S_v$, thus $B$ contains exactly $\Delta_G \cdot |V|$ vertices.
The set of edges $F$ is constructed as follows. 
For each $w \in S_v$ we add to $F$ an edge $vw$.
Additionally, for each $v_e \in B$, where $e \in E, e=uv$ we add to $F$ an edge $uv_e$, so that the vertex $v_e$ is of degree exactly two in $G'$.
The described transformation is depicted in Fig.~\ref{fig:transformation}.
To finish the construction of $I'$ we define $k' = k + \Delta_G \frac{n(n-1)}{2} - 2|E|$.

 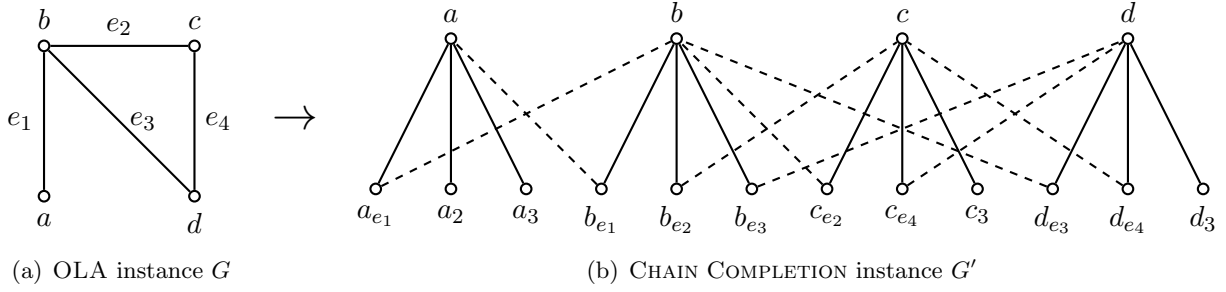
\begin{figure}[h]
  \subfigure[{\olashort instance $G$~~~~~~~~~}]{
\begin{tikzpicture}[thick,scale=1]\label{fig:graphG}
\node[draw,circle,fill=white,minimum size=4pt,
                            inner sep=0pt, label=above:{{$c$}}] (ur) at (1,1.5) {};
\node[draw,circle,fill=white,minimum size=4pt,
                            inner sep=0pt,label=above:{$b$}] (ul) at (-1,1.5) {} edge[-] node[above]{$e_2$}(ur) ;
                            
\node[draw,circle,fill=white,minimum size=4pt,
                            inner sep=0pt,label=below:{$a$}] (bl) at (-1,-0.5) {} edge[-]node[left]{$e_1$}(ul);
\node[draw,circle,fill=white,minimum size=4pt,
                            inner sep=0pt, label=below:{$d$}] (br) at (1,-0.5) {} edge[-]node[above, right]{$e_3$}(ul) edge[-]node[right]{$e_4$}(ur);
\node[label=right:{\LARGE$\rightarrow$}] (arrow) at (1.75, 0.5){};
\end{tikzpicture}}
\subfigure[\chaincompletion instance $G'$]{
\begin{tikzpicture}[thick,scale=1]\label{fig:graphG'}
\node[draw,circle,fill=white,minimum size=4pt,
                            inner sep=0pt, label=above:{{$a$}}] (a) at (-6,1) {};
\node[draw,circle,fill=white,minimum size=4pt,
                            inner sep=0pt,label=above:{$b$}] (b) at (-3,1) {};                           
\node[draw,circle,fill=white,minimum size=4pt,
                            inner sep=0pt,label=above:{$c$}] (c) at (0,1) {} ;
\node[draw,circle,fill=white,minimum size=4pt,
                            inner sep=0pt, label=above:{$d$}] (d) at (3,1) {};
                            
\node[draw,circle,fill=white,minimum size=4pt,
                            inner sep=0pt,label=below:{$a_{e_1}$}] (vab) at (-7,-1) {} edge[-](a) edge[dashed](b);                           
\node[draw,circle,fill=white,minimum size=4pt,
                            inner sep=0pt,label=below:{$a_2$}] (a2) at (-6,-1) {} edge[-](a);
\node[draw,circle,fill=white,minimum size=4pt,
                            inner sep=0pt, label=below:{$a_3$}] (a3) at (-5,-1) {} edge[-](a);
                            
\node[draw,circle,fill=white,minimum size=4pt,
                            inner sep=0pt,label=below:{$b_{e_1}$}] (bba) at (-4,-1) {} edge[](b) edge[dashed](a);                           
\node[draw,circle,fill=white,minimum size=4pt,
                            inner sep=0pt,label=below:{$b_{e_2}$}] (bbc) at (-3,-1) {} edge[-](b) edge[dashed](c);
\node[draw,circle,fill=white,minimum size=4pt,
                            inner sep=0pt, label=below:{$b_{e_3}$}] (bbd) at (-2,-1) {} edge[-](b) edge[dashed](d);
                           
\node[draw,circle,fill=white,minimum size=4pt,
                            inner sep=0pt,label=below:{$c_{e_2}$}] (ccb) at (-1,-1) {} edge[-](c) edge[dashed](b);                           
\node[draw,circle,fill=white,minimum size=4pt,
                            inner sep=0pt,label=below:{$c_{e_4}$}] (ccd) at (0,-1) {} edge[-](c) edge[dashed](d);
\node[draw,circle,fill=white,minimum size=4pt,
                            inner sep=0pt, label=below:{$c_3$}] (c3) at (1,-1) {} edge[-](c);

\node[draw,circle,fill=white,minimum size=4pt,
                            inner sep=0pt,label=below:{$d_{e_3}$}] (ddb) at (2,-1) {} edge[-](d) edge[dashed](b);                           
\node[draw,circle,fill=white,minimum size=4pt,
                            inner sep=0pt,label=below:{$d_{e_4}$}] (ddc) at (3,-1) {} edge[-](d) edge[dashed](c);
\node[draw,circle,fill=white,minimum size=4pt,
                            inner sep=0pt, label=below:{$d_3$}] (d3) at (4,-1) {} edge[-](d);

\end{tikzpicture}}
\caption{Transformation of \ola to \chaincompletion}\label{fig:transformation}


\end{figure}

For a given ordering $\pi$ of the vertices of the graph $G=(V,E)$ denote by $C(G, \pi)$ the cost of arrangement induced by this ordering (i.e., $C(\pi,G)=\sum_{uv\in E} |\pi(u)-\pi(v)|$). For an ordering $\sigma$ of the left side $A$ of the bipartite graph $G'=(A, B, F)$ denote by $E(G',\sigma)$ the number of edges that we should add to obtain a minimal chain graph in which the left order coincides with $\sigma$. We will prove the following claim.


\begin{claim}
\label{claim:chain-reduction}
For any ordering $\pi$ of $V$ (or equivalently $A$) we have $E(G',\pi)=C(G,\pi)+\Delta_G\frac{n(n-1)}{2}-2|E|$. 
\end{claim}

\begin{proof}
Let us inspect what is the number of edges in a minimal chain bipartite graph $G''$ which has left order $\pi=(v_1, v_2, \dots, v_n)$ and contains $G'=(A, B, F)$ as a subgraph. We know that $N(v_i)\subseteq N(v_j)$ for any $i<j$, which means that each $x \in S_{v_i}$ must be connected to all vertices from the set $\{v_i,v_{i+1}, \dots, v_n\}$.
So it means that each vertex from $S_{v_i}$ is connected to at least $(n+1)-i$ vertices from the set $A$.
Moreover, by minimality of $G''$ the vertex $x$ is connected to exactly this number of the vertices if $x$ does not correspond to any edge in $G$,
i.e., when $x$ is of degree exactly one in $G'$.
If a vertex $w_i \in S_{v_i}$ corresponds to some edge $e \in G$ with endpoints $v_i,v_j$ (note there could be several edges with equal endpoints as we are working with multigraphs)
then the vertex $w_i$ is connected to vertices $v_i, v_j$ in the graph $G'$. 
Hence, in a minimal chain graph $G''$ the degree of $w_i$ is either $(n+1)-i$ or $(n+1) - j = (n+1)-i+(i-j)$ depending on whether $i<j$ or $i>j$. 
Note that there is a second vertex $w_j \in S_{v_j}$ which also corresponds to the edge $e$. 
The degrees of $w_i$ and $w_j$ in $G''$ both are equal to  $(n+1)-i$ or $(n+1) - j$, depending whether $i<j$ or $i>j$. 
In both cases the sum of degrees $w_i, w_j$ in $G''$ can be written as $((n+1)- i)+ ((n+1) -j) + |i-j|$.
Hence, for each edge $e$ with endpoints $v_i,v_j$ we have additional cost of $|i-j|$.  
Summing up, we infer that the number of edges in $G''$ equals 
$$\Delta_G (\sum_i^n((n+1)-i))+\sum_{v_iv_j\in E}|i-j|=\Delta_G\frac{n(n+1)}{2}+C(G,\pi).$$

The number of added edges equals the number of edges in $G''$ minus the number of edges in $G'$. 
So we add exactly $$(\Delta_G\frac{n(n+1)}{2}+C(G,\pi))-(\Delta_Gn+2|E|)=\Delta_G\frac{n(n-1)}{2}+C(G,\pi)-2|E|$$
edges.
\end{proof}

Equivalence of the instances $I$ and $I'$ follows from the claim, and it proves the lemma.

\end{proof}

Having an instance of \chaincompletion we transform it further
to an instance of other completion problems
by simply making $A$ a clique, or by making
both $A$ and $B$ cliques.
By inspecting the forbidden subgraphs definition of each graph class
we infer the equivalence of the instances,
which is enough to prove Theorems~\ref{thm:only-eth} and~\ref{thm:main}.


\begin{figure}[ht]
	\centering
		\subfigure[claw]{ 
			\begin{tikzpicture}
			\node[draw,circle,fill=white,minimum size=4pt,
			inner sep=0pt] (c) at (0:0) {};
			\node[draw,circle,fill=white,minimum size=4pt,
			inner sep=0pt] (r1) at (-30:1) {} edge[-](c) ;
			\node[draw,circle,fill=white,minimum size=4pt,
			inner sep=0pt] (u1) at (90:1) {} edge[-](c);
			\node[draw,circle,fill=white,minimum size=4pt,
			inner sep=0pt] (l1) at (210:1) {} edge[-](c);
			\end{tikzpicture}}
		\hspace{.04\textwidth}
		\subfigure[$P_4$]{ 
			\begin{tikzpicture}
			\node[draw,circle,fill=white,minimum size=4pt,
				inner sep=0pt] (v1) at (-1.5,0) {};
				\node[draw,circle,fill=white,minimum size=4pt,
				inner sep=0pt] (v2) at (-0.5,0) {} edge[-](v1) ;
				\node[draw,circle,fill=white,minimum size=4pt,
				inner sep=0pt] (v3) at (0.5,0) {} edge[-](v2);
				\node[draw,circle,fill=white,minimum size=4pt,
				inner sep=0pt] (v4) at (1.5,0) {} edge[-](v3);
				\end{tikzpicture}}
			\hspace{.04\textwidth}
	\subfigure[$2 K_2$]{ 
		\begin{tikzpicture}
		\node[draw,circle,fill=white,minimum size=4pt,
		inner sep=0pt] (v1) at (-0.5,-0.5) {};
		\node[draw,circle,fill=white,minimum size=4pt,
		inner sep=0pt] (v2) at (-0.5,0.5) {} edge[-](v1) ;
	    \node[draw,circle,fill=white,minimum size=4pt,
		inner sep=0pt] (v3) at (0.5,-0.5) {} ;
		\node[draw,circle,fill=white,minimum size=4pt,
		inner sep=0pt] (v4) at (0.5,0.5) {} edge[-](v3);
		\end{tikzpicture}}
	\hspace{.04\textwidth}
	\subfigure[$C_n$, $n\geq 4$]{
		\begin{tikzpicture}
		\node[draw,circle,fill=white,minimum size=4pt,
		inner sep=0pt, label=above:{\tiny{$n$}}] (ur) at (1,1.5) {};
		\node[draw,circle,fill=white,minimum size=4pt,
		inner sep=0pt,label=above:{\tiny{$1$}}] (ul) at (-1,1.5) {} edge[-](ur) ;
		
		\node[draw,circle,fill=white,minimum size=4pt,
		inner sep=0pt,label=below:{\tiny{$2$}}] (bl) at (-1,-0.5) {} edge[-](ul);
		\node[draw,circle,fill=white,minimum size=4pt,
		inner sep=0pt, label=below:{\tiny{$\dots$}}] (br) at (1,-0.5) {} edge[-](bl) edge[-](ur);
		\end{tikzpicture}}
		\hspace{.04\textwidth}
	\subfigure[bipartite claw]{ 
		\begin{tikzpicture}
		\node[draw,circle,fill=white,minimum size=4pt,
		inner sep=0pt] (c) at (0:0) {};
		\node[draw,circle,fill=white,minimum size=4pt,
		inner sep=0pt] (r1) at (-30:1) {} edge[-](c) ;
		\node[draw,circle,fill=white,minimum size=4pt,
		inner sep=0pt] (r2) at (-30:2) {} edge[-](r1);
		\node[draw,circle,fill=white,minimum size=4pt,
		inner sep=0pt] (u1) at (90:1) {} edge[-](c);
		\node[draw,circle,fill=white,minimum size=4pt,
		inner sep=0pt] (u2) at (90:2) {} edge[-](u1);
		\node[draw,circle,fill=white,minimum size=4pt,
		inner sep=0pt] (l1) at (210:1) {} edge[-](c);
		\node[draw,circle,fill=white,minimum size=4pt,
		inner sep=0pt] (l2) at (210:2) {} edge[-](l1);
		\end{tikzpicture}}
	\hspace{.04\textwidth}
	\subfigure[umbrella]{
		\begin{tikzpicture}
		\node[draw,circle,fill=white,minimum size=4pt,
		inner sep=0pt] (c) at (0,1) {};
		\node[draw,circle,fill=white,minimum size=4pt,
		inner sep=0pt] (u) at (0,2) {} edge[-](c) ;
		\node[draw,circle,fill=white,minimum size=4pt,
		inner sep=0pt] (r1) at (1,1) {} edge[-](c) edge[-](u);
		\node[draw,circle,fill=white,minimum size=4pt,
		inner sep=0pt] (r2) at (2,1) {} edge[-](u) edge[-](r1);
		\node[draw,circle,fill=white,minimum size=4pt,
		inner sep=0pt] (l1) at (-1,1) {} edge[-](c) edge[-](u);
		\node[draw,circle,fill=white,minimum size=4pt,
		inner sep=0pt] (r1) at (-2,1) {} edge[-](u) edge[-](l1);
		\node[draw, circle,fill=white,minimum size=4pt,inner sep=0pt] (b) at (0,-1){} edge[-](c);
		\end{tikzpicture}}
	\hspace{.04\textwidth}
	\subfigure[$n$-net, $n\geq 2$]{
		\begin{tikzpicture}
		\node[draw,circle,fill=white,minimum size=4pt,
		inner sep=0pt] (u2) at (0,2) {};
		\node[draw,circle,fill=white,minimum size=4pt,
		inner sep=0pt] (u1) at (0,1) {} edge[-](u2) ;
		\node[draw,circle,fill=white,minimum size=4pt,
		inner sep=0pt,label=below:{\tiny{$1$}}] (l2) at (-1.5,-0.5) {} edge[-](u1);
		\node[draw,circle,fill=white,minimum size=4pt,
		inner sep=0pt, label=below:{\tiny{$2$}}] (l1) at (-0.6,-0.5) {} edge[-](u1) edge[-](l2);
		
		\node[draw,circle,fill=white,minimum size=4pt,
		inner sep=0pt, label=below:{\tiny{$n$}}] (r2) at (1.5,-0.5) {} edge[-](u1)  ;
		\node[draw,circle,fill=white,minimum size=4pt,
		inner sep=0pt, label=below:{\tiny{$\dots$}}] (r1) at (0.6,-0.5) {} edge[-](u1) edge[-](r2) edge[-](l1);
		
		\node[draw,circle,fill=white,minimum size=4pt,
		inner sep=0pt] (ll) at (-2,-1) {} edge[-](l2);
		\node[draw,circle,fill=white,minimum size=4pt,
		inner sep=0pt] (rr) at (2,-1) {} edge[-](r2);
		\end{tikzpicture}}
	\hspace{.04\textwidth}
	\subfigure[$n$-tent, $n\geq 3$]{
		\begin{tikzpicture}
		\node[draw,circle,fill=white,minimum size=4pt,
		inner sep=0pt] (u) at (0,2) {};
		
		\node[draw,circle,fill=white,minimum size=4pt,
		inner sep=0pt] (l) at (-0.5,0.7) {} edge[-](u) ;
		\node[draw,circle,fill=white,minimum size=4pt,
		inner sep=0pt] (r) at (0.5,0.7) {} edge[-](u) edge[-](l);
		
		\node[draw,circle,fill=white,minimum size=4pt,
		inner sep=0pt,label=below:{\tiny{$1$}} ] (b1) at (-2,-0.8) {} edge[-](l);
		\node[draw,circle,fill=white,minimum size=4pt,
		inner sep=0pt, label=below:{\tiny{$2$}}] (b2) at (-1,-0.8) {} edge[-](b1) edge[-](l) edge[-](r) ;
		\node[draw,circle,fill=white,minimum size=4pt,
		inner sep=0pt, label=below:{\tiny{$3$}}] (b3) at (0,-0.8) {} edge[-](b2) edge[-](l) edge[-](r);
		\node[draw,circle,fill=white,minimum size=4pt,
		inner sep=0pt,label=below:{\tiny{$\dots$}}] (b4) at (1,-0.8) {} edge[-](b3) edge[-](l) edge[-](r);
		\node[draw,circle,fill=white,minimum size=4pt,
		inner sep=0pt,label=below:{\tiny{$n$}}] (b5) at (2,-0.8) {} edge[-](r) edge[-](b4);
		\end{tikzpicture}}
	\caption[]{Forbidden induced subgraphs for various graph classes}\label{fig:interval}
\end{figure}
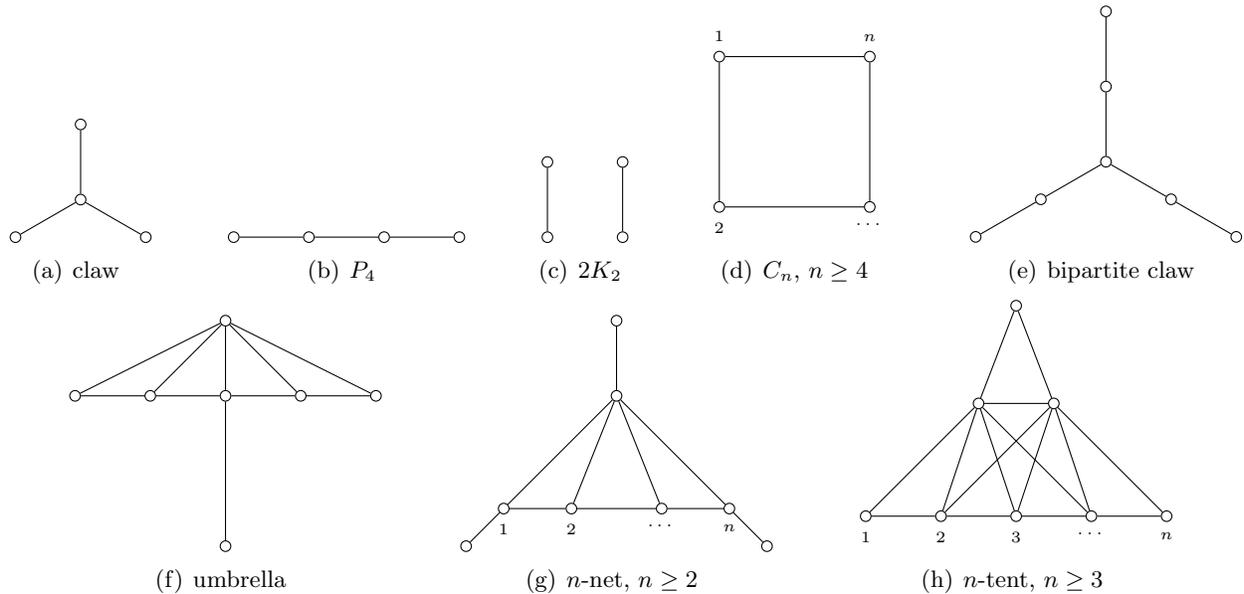

Classes of chordal, interval, proper interval, threshold, trivially perfect graphs have many characterizations. For our purposes the most convenient one is by the set of forbidden induced subgraphs \cite{brandstadt}. The characterization is presented in Table~1.

\begin{table}
	\begin{center}
		\begin{tabular}{|c|c|}
			\hline
			~~Graph class name~~   & ~~Forbidden induced subgraphs~~     \\
			\hline
			Chordal & $C_n$ for $n \geq 4$             \\
			Proper Interval & claw, $2$-net, $3$-tent, $C_n$ for $n\geq 4$\\
			Interval & bipartite claw, umbrella, $n$-net for $n\geq 2$, $n$-tent for $n \geq 3$, $C_n$ for $n\geq 4$ \\
			Threshold & $2K_2, C_4, P_4$ \\
			Trivially perfect & $C_4, P_4$ \\
			\hline
		\end{tabular}
	\end{center}
	\caption{Characterization of graph classes by forbidden induced subgraphs}
	\label{tab:forbidden}
\end{table}

\begin{lemma}\label{lem:chain-to-compl}
There are polynomial time reductions from \chaincompletion problem to \minimumfillin, \intervalcompletion, \properintervalcompletion, \thresholdcompletion, \triviallyperfectcompletion problems and these reductions do not change vertex set.

\end{lemma}

\begin{proof}

For any bipartite graph $H'=(U_1,U_2,F)$ consider a graph $Ch(H')=(U_1,U_2,F\cup \{uv| u,v \in U_1\}\cup \{uv|u,v \in U_2\})$. In \cite{yannakakis}  it is shown that any bipartite graph $H'$ is a chain graph if and only if $Ch(H')$ is a chordal graph. It means that in order to reduce an instance of \chaincompletion $H=(A,B,F)$ to \minimumfillin it is enough to construct cliques on sets of vertices $A$ and $B$.
So the constructed graph is a union of two cliques and some edges between cliques. Hence any arbitrarily completion of this graph does not contain a claw, bipartite claw, umbrella, $p$-net, $q$-tent for $p \geq 2, q\geq3$ as these graphs have an independent set of size $3$ and a union of two cliques does not. It follows that solutions for \minimumfillin, \properintervalcompletion,\intervalcompletion problems on such instances coincide and we can look at this reduction to \minimumfillin as a reduction to \properintervalcompletion or \intervalcompletion problems.

It is left to show a reduction from a \chaincompletion instance to \thresholdcompletion and \triviallyperfectcompletion instances. Having a \chaincompletion instance with a bipartite graph $H=(A,B,F)$ we consider \triviallyperfectcompletion and \thresholdcompletion problems on the graph $G=(A\cup B, F \cup \{(u,v)| u,v \in A\})$.  We just add edges such that $A$ becomes a clique. We show that a minimum chain completion of the graph $H$ corresponds to a completion towards trivially perfect or threshold graph. Let $F'$ be a solution of \chaincompletion for the graph $H$. Consider graph $G'=(A\cup B, F \cup \{(u,v)| u,v \in A\} \cup F')$,  $G'$ is a union of independent set and clique plus some edges between independent set and clique. So it does not contain induced $2K_2$ or $C_4$. If $G'$ contains induced $P_4=v_1v_2v_3v_4$ then $v_2,v_3$ belong to clique and $v_1, v_4$ to independent set. However, this contradict to fact that edges between clique and independent set form a chain completion.   As $G'$ does not contain induced $2K_2, P_4, C_4$ it is trivially perfect and threshold graph. Let now $F'$ denote solution of \triviallyperfectcompletion or \thresholdcompletion on instance $G'$. To finish the proof of correctness of reduction it is enough to show that $(A, B, (F\cup F') \cap E(A,B))$ is a chain graph. If this graph is not a chain graph then it must contain two independent edges $v_1v_2, v_3v_4$ \cite{yannakakis}. However in such case graph on vertices $v_1,v_2,v_3,v_4$ will induce a $P_4$ in graph $G' \cup F'$ which contradict to the fact that $G'\cup F'$ is a threshold or trivially perfect graph.

To sum up all our reductions from \chaincompletion to \minimumfillin, \properintervalcompletion, \intervalcompletion, \thresholdcompletion, \triviallyperfectcompletion  add some edges to a graph of \chaincompletion instance and do not change a vertex set of a graph.
\end{proof}

At this point we almost proved Theorem~\ref{thm:only-eth} and Theorem~\ref{thm:main}.
\begin{reptheorem}{thm:only-eth}
	Unless ETH fails, there is an integer $c\geq 1$ such that there are no $2^{\Oh(\sqrt{n}/\log^c n)}$, and consequently no $2^{\Oh(k^{1/4}/\log^c k)}\cdot n^{\Oh(1)}$ algorithms for the following problems: \minimumfillin, \intervalcompletion, \properintervalcompletion, \triviallyperfectcompletion, \thresholdcompletion, {\sc{Chain Completion}}.
\end{reptheorem}
\begin{proof}
If the statement is not true then for some of the problems there is an algorithm running in time $2^{\Oh(\sqrt{n}/\log^c n)}$. Having instance of \ola on $n$ vertices we can reduce it to problem under consideration with $(\Delta_G+1)n=\Oh(n^2)$ vertices by Lemmas~\ref{lem:ola-to-chain} and \ref{lem:chain-to-compl}. This gives us an  $2^{\Oh(\sqrt{n^2}/\log^c n^2)}=2^{\Oh(n/log^c n)}$ time  algorithm which contradicts Theorem~\ref{thm:eth-ola}. As $k \leq n^2$ we also have  $2^{\Omega(k^{1/4}/\log^c k)}\cdot n^{\Oh(1)}$ lower bound on the running time.
\end{proof}

\begin{reptheorem}{thm:main}
	Unless Hypothesis~\ref{h.bisection} fails, there is no $2^{o(n+m)}$-time algorithm for {\sc{Chain Completion}}, and no $2^{o(n)}$-time algorithms for \minimumfillin, \intervalcompletion, \properintervalcompletion, \triviallyperfectcompletion, and \thresholdcompletion. Consequently, none of these problems can be solved in time $2^{o(\sqrt{k})}\cdot n^{\Oh(1)}$.
\end{reptheorem}
\begin{proof}
	In Section~\ref{sec:sparse} we  transformed a $d$-regular \minb instance to \ola instance with bounded degree. Pipelined with Lemma~\ref{lem:ola-to-chain} we get a reduction from $d$-regular \minb to a \chaincompletion instance with $\Oh(n)$ vertices and edges. So $2^{o(n+m)}$-time algorithm for \chaincompletion  contradicts Hypothesis~\ref{h.bisection}. By Lemma~\ref{lem:chain-to-compl} we can reduce \chaincompletion  to \minimumfillin, \intervalcompletion, \properintervalcompletion, \triviallyperfectcompletion, and \thresholdcompletion instance without changing the vertex set. Combining all three reductions in one we get reductions from \minb to \minimumfillin, \properintervalcompletion, \intervalcompletion, \thresholdcompletion, \triviallyperfectcompletion problems which transform an instance with $n$ vertices into an instance with $\Oh(n)$ vertices. This leads to $2^{\Omega(n)}$ lower bound for all discussed completion problems as well as to $2^{\Omega(\sqrt{k})}\cdot n^{\Oh(1)}$ lower bound because $k \leq n^2$.	   
\end{proof}

%% file: sec_fast.tex
\section{Hardness of \fast}
\label{sec:fast}

In this section we prove Theorem~\ref{thm:only-eth-fast}, that is, the lower bound on the complexity of \fast. We start with preparing an appropriately hard instance of \fas in general digraphs, so that we can apply the reduction of Ailon et al.~\cite{AilonCN08}.

\subsection{Preparing a hard instance of \fasshort}

By $\ssat[1,1]{d,d}$ we denote the version of $\sat[3]$ where every variable has 
\begin{itemize}
\item exactly $1$ positive occurrence in a clause of size $3$,
\item exactly $1$ negative occurrence in a clause of size $3$,
\item exactly $d$ positive occurrences in clauses of size $2$,
\item exactly $d$ negative occurrences in clauses of size $2$, and
\item there are no clauses of size $1$.
\end{itemize}
Similarly as before, $\gapssat[1,1]{d,d}_{[\bounda, \boundb]}$ for $0\leq \bounda<\boundb\leq 1$ is the problem of distinguishing whether the maximum number of clauses that can be satisfied in a given instance of $\ssat[1,1]{d,d}$ is at most $\bounda m$, or at least $\boundb m$, where $m$ is the total number of clauses. We now give a hardness result for $\gapssat[1,1]{d,d}_{[\bounda, \boundb]}$.

\begin{lemma}\label{lem:gapssat}
There exists a positive integer $d$ such that
 $$\gapenaesat[3]_{[\bounda, 1]} \linleq \gapssat[1,1]{d,d}_{[\frac{1+\bounda+3d}{2+3d}, 1]}$$
\end{lemma}
\begin{proof}
Let $\phi$ be the input E3-CNF formula. For every variable $x$, let $n(x)$ be the number of occurrences of $x$ in $\phi$. Let $G_x$ be a $d$-regular expander with $n(x)$ vertices, obtained using Theorem~\ref{thm:expander-factory} for $p=2$; suppose vertices of $G_x$ are numbered $1,2,\ldots,n(x)$. Create $\psi$ as follows: 
\begin{itemize}
\item Replace $x$ with $n(x)$ new variables $x_1,x_2,\ldots,x_{n(x)}$, corresponding to the occurrences of $x$ in $\phi$;
\item For every $ij\in E(G_x)$ with $i\neq j$, introduce two clauses $(\neg x_i\vee x_j)$ and $(x_i\vee \neg x_j)$, which (if unbroken) force the evaluation of $x_i$ to be equal to that of $x_j$;
\item For every $ii\in E(G_x)$, introduce a trivial clause $(\neg x_i\vee x_i)$;
\item Perform the same construction for every other variable of $\phi$;
\item For every $3$-clause $C$ of $\phi$, introduce two new clauses $C'$ and $C''$ to $\psi$. Clause $C'$ is constructed from $C$ by replacing each occurrence of some variable $x$ with the new variable $x_i$ corresponding to this occurrence; the polarity of the literals in $C'$ is the same as in $C$. Clause $C''$ is constructed from $C'$ by reversing the polarity of each literal. For example, if $C=(x\vee \neg y\vee z)$, then $C'=(x_i\vee \neg y_j\vee z_k)$ and $C''=(\neg x_i\vee y_j\vee \neg z_k)$, where $i,j,k$ are the indices of the occurrences of $x,y,z$ in $C$, respectively.
\end{itemize}
It is easy to see that every variable of the new formula $\psi$ appears in exactly two $3$-clauses, once positively and once negatively. Also, it has exactly $2d$ occurrences in $2$-clauses: $d$ positive and $d$ negative. Since the original formula had only clauses of size $3$, the new formula has exactly $3m$ variables and $m(2+3d)$ clauses. 

If there is a variable evaluation $\lambda$ for $\phi$ that NAE-satisfies all clauses of $\phi$, then we can construct a variable evaluation $\lambda'$ for $\psi$ by assigning all the variables $x_i$ that originate in variables $x$ the value $\lambda(x)$. Then it is easy to see that $\lambda'$ satisfies all the clauses of $\psi$.

Suppose now that every variable evaluation for $\phi$ NAE-satisfies at most $\bounda m$ clauses, and for the sake of contradiction suppose that there is a variable evaluation $\lambda'$ for $\psi$ that satisfies more than $(1+\bounda+3d)m$ clauses. Let $x$ be a variable of $\phi$, and let us modify $\lambda'$ as follows: assign all the variables $x_1,x_2,\ldots,x_{n(x)}$ the value that is taken by the majority of these variables in the original evaluation $\lambda'$ (breaking ties arbitrarily). Observe that since $h(G_x)\geq 2$, this step cannot decrease the number of satisfied clauses: if $q$ is the number of variables out of $x_1,x_2,\ldots,x_{n(x)}$ that take the minority value, then by replacing their values by the majority value we can unsatisfy at most $2q$ $3$-clauses, but we satisfy at least $h(G_x)\cdot q\geq 2q$ $2$-clauses that were previously unsatisfied. By performing this operation for every variable of $\phi$, we can assume without loss of generality that in $\lambda'$ all the variables originating in the same variable of $\phi$ are assigned the same value. This naturally defines a variable evaluation $\lambda$ for $\phi$. Then, provided $\lambda'$ satisfied more than $(1+\bounda+3d)m$ clauses of $\psi$, we infer that $\lambda$ NAE-satisfies more than $\bounda m$ clauses of $\phi$. This is a contradiction.
\end{proof}

We now turn our attention to the $\fvs$ in general directed multigraphs: given a directed multigraph $G$, find the smallest possible subset of vertices $X$ such that $G-X$ is acyclic. By $\fvss[d]$ we denote the variant of $\fvsshort$ where the input directed multigraph has no loops, and is $2d$-regular and balanced, i.e., the indegree and the outdegree of every vertex is equal to $d$. Again, $\gapfvs[d]_{[\bounda, \boundb]}$ is the gap problem where we need to distinguish between the cases when the optimum size of $X$ is at most $\bounda n$ and at least $\boundb n$, where $n$ is the number of vertices of the multigraph.

\begin{lemma}\label{lem:fvs}
 $$\gapssat[1,1]{d,d}_{[\bounda, 1]} \linleq \gapfvs[d+2]_{[1/2,(4-\bounda)/6]}$$
\end{lemma}
\begin{proof}
Let $\phi$ be the input instance of $\ssat[1,1]{d,d}$, and let $n$ be the number of variables of $\phi$. Then $m$, the number of clauses of $\phi$, is equal to $(2/3+d)n$.

Construct a directed multigraph $G$ as follows. For every variable $x$ of $\phi$, create two vertices $u^x_{\top}$ and $u^x_{\bot}$, corresponding to setting $x$ to true and false, respectively. Add edges $(u^x_{\top},u^x_{\bot})$ and $(u^x_{\bot},u^x_{\top})$ to the edge set, for every variable $x$ of $\phi$. Moreover, for every $2$-clause of $\phi$ add a $2$-cycle between vertices corresponding to its literals (e.g. clause $x\vee \neg y$ gives rise to edges $(u^x_{\top},u^y_{\bot})$ and $(u^y_{\bot},u^x_{\top})$), and similarly for every $3$-clause of $\phi$ add a $3$-cycle between vertices corresponding to its literals, oriented arbitrarily. Note that in this manner trivial clauses of the form $(\neg x\vee x)$ give rise to additional copies of the $2$-cycle $(u^x_{\top},u^x_{\bot})(u^x_{\top},u^x_{\bot})$. This concludes the construction of $G$. It is easy to verify using the assumed properties of $\phi$ that every vertex of $G$ has indegree and outdegree equal to $d+2$, and moreover in the construction we did not introduce loops. Let $n'=2n$ be the number of vertices in $G$.

Suppose first that there exists a variable evaluation $\lambda$ for $\phi$ that satisfies all the clauses of $\phi$. Define $X$ to be the set of all the vertices $u^x_{\lambda(x)}$ for $x$ being a variable of $\phi$; note that $|X|=n=n'/2$. Since $\lambda$ satisfies all the clauses, and every variable of $\phi$ participates in exactly one $3$-clause positively and in exactly one $3$-clause negatively, then it is easy to see that all the weakly connected components of $G-X$ are either isolated vertices or single edges. Thus, $G-X$ is acyclic.

Assume now that every variable evaluation for $\phi$ satisfies at most $\bounda m$ clauses, and for the sake of contradiction suppose that there exists a set $X$ with $|X|<\frac{4-\bounda}{6}n'=\frac{4-\bounda}{3}n$ such that $G-X$ is acyclic. Observe that from each pair $\{u^x_{\top},u^x_{\bot}\}$ at least one vertex has to belong to $X$. Define a variable evaluation $\lambda$ for $\phi$ as follows: if $|\{u^x_{\top},u^x_{\bot}\}\cap X|=1$, then $\lambda(x)$ is such that $u^x_{\lambda(x)}\in X$, and otherwise $\lambda(x)$ is chosen arbitrarily. Observe that the first alternative holds for a set of more than $n-\frac{1-\bounda}{3}n$ variables; let us denote them by $S$. Since $G-X$ is acyclic, each of $2$- and $3$-cycles constructed for a clause $C$ of $\phi$ has at least one vertex from $X$. If all the variables of $C$ belong to $S$, then it can be easily seen that this implies that $\lambda$ satisfies $C$. Hence, the only clauses of $\phi$ that can be unsatisfied by $\lambda$ are the ones that contain at least one variable outside $S$. Every variable of $\phi$ occurs in at most $2d+2$ clauses, and there are less than $\frac{1-\bounda}{3}\cdot n$ variables outside $S$, which means that $\lambda$ unsatisfies less than $\frac{(1-\bounda)(2d+2)}{3}\cdot n$ clauses. Hence the fraction of unsatisfied clauses is less than
$$\frac{(1-\bounda)(2d+2)n}{3m}=\frac{(1-\bounda)(2d+2)}{(2+3d)}\leq 1-\bounda.$$
This is a contradiction.
\end{proof}

Finally, there is a well-known reduction that reduces \fvs to \fas in the directed setting. This reduction appears to preserve the gap. In the following, by $\gapfas[d]_{[\bounda, \boundb]}$ we denote the problem of determining, for a given directed multigraph $G$ without loops whose underlying undirected multigraph is $d$-regular, whether the minimum number of edges that needs to be removed from $G$ to make it acyclic is at most $\bounda m$ or at least $\boundb m$, where $0\leq \bounda<\boundb\leq 1$ and $m$ is the number of edges in $G$.

\begin{lemma}\label{lem:fas}
$$\gapfvs[d]_{[\bounda, \boundb]} \linleq \gapfas[d+1]_{[\frac{\bounda}{d+1}, \frac{\boundb}{d+1}]}$$
\end{lemma}
\begin{proof}
Let $G$ be the input directed multigraph, and let $n$ and $m$ denote the numbers of edges and vertices of $G$, respectively; by the assumption that $G$ is $2d$-regular we know that $m=dn$. Construct a graph $G'$ as follows:
\begin{itemize}
\item For every $u\in V(G)$ create two vertices $u^-,u^+\in V(G')$ and an edge $(u^-,u^+)\in E(G')$;
\item For every edge $(u,v)\in E(G)$, create an edge $(u^+,v^-)$.
\end{itemize}
This concludes the construction of $G'$. Let $E_1,E_2$ be the sets of edges constructed in the first and second bullet point, respectively. Since $G$ was $2d$-regular and balanced, we infer that every vertex $u^+$ has outdegree $d$ and indegree $1$, whereas every vertex $u^-$ has outdegree $1$ and indegree $d$. Thus, $G'$ has $n'=2n$ vertices and $m'=m+n=(d+1)n$ edges.

Suppose first that $X$ is a subset of vertices of $G$ with size at most $\bounda n$ such that $G-X$ is acyclic. Let $F=\{(u^+,u^-)\,|,u\in X\}\subseteq E_1$. Then it can be easily seen that $G'-F$ is acyclic, and $|F|=|X|\leq \bounda n=\frac{\bounda}{d+1}m'$.

Assume now that every subset $X\subseteq V(G)$ for which $G-X$ is acyclic has size at least $\boundb n$, and for the sake of contradiction suppose that there is a set $F\subseteq E(G')$ such that $G'-F$ is acyclic and $|F|<\frac{\boundb}{d+1}m'$. Observe that if $F$ contains some edge $(u^+,v^-)\in E_2$, then we could modify $F$ by removing $(u^+,v^-)$ from $F$ and adding $(v^-,v^+)$ to $F$ (unless it is not already contained in $F$, in which case we do not add any edge to $F$). This operation can only decrease the number of edges in $F$ and preserves the property that $G'-F$ is acyclic; this is because after removing $(v^-,v^+)$, $v^-$ becomes a sink. Thus, without loss of generality we can assume that $F\subseteq E_1$. Let $X$ be the set of vertices $u\in V(G)$ for which $(u^-,u^+)\in F$. Since $G'-F$ is acyclic, it easily follows that $G-X$ is also acyclic. Moreover, $|X|=|F|<\frac{\boundb}{d+1}m'=\boundb n$. This is a contradiction. 
\end{proof}

Finally, observe that in an instance of $\fasshort$ without loops one can subdivide every edge once, which doubles the number of edges while not changing the size of the optimum solution. Thus, application of this reduction to the gap problem shrinks the gap twice and makes the directed graph at hand simple: it has no loops, no parallel edges, and moreover if $(u,v)$ is an edge then $(v,u)$ is not. By combining this observation with Theorem~\ref{thm:hard-gapsat} and Lemmas~\ref{lem:3sat-nae4sat},~\ref{lem:nae4sat-nae3sat},~\ref{lem:gapssat},~\ref{lem:fvs}, and~\ref{lem:fas}, we obtain the following result.

\begin{theorem}\label{thm:fas-hardness}
Unless ETH fails, there exist $0\leq\bounda<\boundb\leq 1$, $c \geq 1$, and $d>0$ such that there is no $2^{\mathcal{O}(\frac{n}{\log^c(n)})}$ algorithm for $\justgapfas_{[\bounda,\boundb]}$ on directed simple graphs of maximum total degree $d$.
\end{theorem}

\subsection{Reducing \fasshort to \fastshort}

\newcommand{\fasv}{\textrm{fas}}
\newcommand{\Exp}{\mathbb{E}}

\begin{reptheorem}{thm:only-eth-fast}
Unless ETH fails, there is an integer $c\geq 1$ such that there is no $2^{\Oh(\sqrt{n}/\log^c n)}$, and consequently no $2^{\Oh(k^{1/4}/\log^c k)}\cdot n^{\Oh(1)}$ algorithm for \fast.
\end{reptheorem}
\begin{proof}
We provide a randomized reduction that essentially reiterates the argument of Ailon et al.~\cite{AilonCN08}.  In the analysis, we use the known fact that for any directed graph $H$, $\fasv(H)$ is equal to the minimum possible number of edges oriented backwards (called {\em{feedback edges}}) in an ordering of vertices of $V(H)$, where $\fasv(H)$ is the optimum size of a feedback arc set in $H$. Thus, we may equivalently think of the $\fasshort$ problem as finding an ordering $\pi$ of $V(H)$ that minimizes the number of feedback edges. For an ordering $\pi$ of $V(H)$, by $\fasv(H,\pi)$ we denote the number of feedback edges in the ordering $\pi$.

Let $\bounda,\boundb,c,d$ be the constants given by Theorem~\ref{thm:fas-hardness}, and let $G$ be an instance of $\justgapfas_{[\bounda,\boundb]}$, where $G$ is a simple directed
graph with $n$ vertices and $m$ edges and has maximum total degree $d$. Note that w.l.o.g. we may assume $m \ge n$, as otherwise
there is a vertex in $G$ with no outgoing edges, which can be safely removed.
Let us fix an integer $k$, to be determined later. We consider the {\em{$k$-blow up}} $G_k$ defined as follows: for every $u\in V(G)$ we create a $k$ vertices $u_1,u_2,\ldots,u_k$ in $G_k$, and for all $1\leq i,j\leq k$ we put $(u_i,v_j)\in E(G_k)$ if and only if $(u,v)\in E(G)$. Thus, vertices $u_i$ are twins. Ailon et al.~\cite{AilonCN08}, based on a communication by Alon, argue that there is an optimum ordering for $G_k$ which may be obtained by taking an optimum ordering for $G$ and replacing every vertex $u\in V(G)$ by a block consisting of vertices $\{u_i\}_{1\leq i\leq k}$ in any order; hence in particular $\fasv(G_k)=k^2\cdot \fasv(G)$. For an ordering $\sigma$ of $V(G)$, let $\sigma_k$ be an ordering of $V(G_k)$ constructed in the manner descibed above.

Construct a tournament $T_k$ from $G_k$ by adding edges between every pair of vertices that are not connected by an edge in $G_k$, where the orientations of these edges are chosen independently and uniformly at random. Observe that $|E(G_k)|=k^2\cdot |E(G)|\leq dk^2n/2$. Let $R_k=(V(T_k),E(T_k)\setminus E(G_k))$ be the directed graph consisting only of the edges picked at random. Then for a sufficiently large $n$ we have that $|E(R_k)|=\binom{nk}{2}-|E(G_k)|\geq \frac{(nk)^2}{4}$, and of course $|E(R_k)|\leq \frac{(nk)^2}{2}$.

We now prove that with high probability, $\fas(T_k)$ is closely related to \newline $\fas(G_k)$, because the number of feedback edges that need to be chosen from the edges picked at random is concentrated around the expected value. 

Let us fix some ordering $\pi$ of $V(G_k)$, then $\fasv(T_k,\pi)=\fasv(G_k,\pi)+\fasv(R_k,\pi)$. For $e\in E(R_k)$, let $X_e$ be the indicator random variable having value $1$ if $e$ is a feedback edge w.r.t. $\pi$, and $0$ otherwise. Let also $X=\sum_{e\in E(R_k)} X_e$; then $\Exp X = \frac{|E(R_k)|}{2}$. Let $\eta=\frac{\boundb-\bounda}{3}$. Since $X_e$-s are independent, from the Chernoff bound we obtain that
$$\Pr(|X-\Exp X| \geq \eta k^2 n)\leq 2\exp\left(-\frac{2\eta^2k^4n^2}{|E(R_k)|}\right)\leq 2\exp(-4k^2\eta^2).$$

Suppose now that there exists an ordering $\pi$ of $V(G)$ that has at most $\bounda m$ feedback edges. Then, with probability at least $1-2\exp(-4\eta^2k^2)$ we have that
\begin{equation}\label{e1}
\fasv(T_k,\pi_k)\leq \bounda k^2 m + |E(R_k)|/2 + \eta k^2n\leq \frac{2\bounda+\boundb}{3}\cdot k^2m+|E(R_k)|/2.
\end{equation}
Hence, if we can set $k$ to be a large enough constant, such that conclusion (\ref{e1}) holds with probability at least $3/4$.

Suppose now that $\fasv(G)\geq \boundb m$. Then, for a fixed ordering $\sigma$ of $V(G_k)$ we have that with probability at least $1-2\exp(-4\eta^2k^2)$ it holds that
\begin{eqnarray}\label{e2}
\fasv(T_k,\sigma)&\geq & \fasv(G_k,\sigma)+|E(R_k)|/2 - \eta k^2n\geq k^2\fasv(G)+|E(R_k)|/2 - \eta k^2n\nonumber\\
& \geq & k^2\boundb m+|E(R_k)|/2 - \eta k^2n\geq \frac{\bounda+2\boundb}{3}\cdot k^2m+|E(R_k)|/2.
\end{eqnarray}
We would like to infer that with high probability this conclusion holds for all the possible orderings $\sigma$, and for this we will use the union bound. Observe that the number of orderings $\sigma$ of $V(G_k)$ is $(nk)!=\exp(\mathcal{O}(nk \log (nk)))$, while the probability of failure for each of them is at most $2\exp(-4\eta^2k^2)$. Since $\eta$ is a positive constant, simple computations show that by setting $k=\Theta(n\log n)$, we have that $(nk)!\cdot 2\exp(-4\eta^2k^2)\leq 1/4$, and hence conclusion (\ref{e2}) holds simultaneously for all orderings $\sigma$ with probability at least $3/4$.

Suppose now that \fast admitted an algorithm with running time $2^{\Oh(\frac{n^{1/2}}{\log^{c'}n})}$ for $c'=c+\frac{1}{2}$, where $c$ is as in Theorem~\ref{thm:fas-hardness}. Apply this algorithm to the constructed tournament $T_k$ to compute $\fasv(T_k)$. In case $\fasv(G)\leq \bounda m$, then with probability at least $3/4$ we have that $\fasv(T_k)\leq \frac{2\bounda+\boundb}{3}\cdot k^2m+|E(R_k)|/2$. In case $\fasv(G)\geq \boundb m$, then with probability at least $3/4$ we have that $\fasv(T_k)\geq \frac{\bounda+2\boundb}{3}\cdot k^2m+|E(R_k)|/2$. Since $\frac{2\bounda+\boundb}{3}<\frac{\bounda+2\boundb}{3}$, these two alternatives are disjoint and the algorithm can, with double-sided error, resolve the input instance of $\justgapfas_{[\bounda,\boundb]}$. Since $|V(T_k)|=\Theta(n^2\log n)$, this procedure runs in time $2^{\Oh(\frac{n}{\log^{c}(n)})}$. This is a contradiction with Theorem~\ref{thm:fas-hardness}.
\end{proof}

%% file: sec_conclusions.tex
\section{Conclusions}

In this work we have given evidence that $2^{\Oh(\sqrt{k}\cdot \textrm{polylog}(k))}\cdot n^{\Oh(1)}$ can be the final answer for the running times of parameterized algorithms for \minimumfillin, \intervalcompletion, \properintervalcompletion, \triviallyperfectcompletion, \thresholdcompletion, and \chaincompletion. This evidence is based on a new complexity hypothesis connected to the hardness of approximation for the \minb problem. Thus, the answer given by us is not completely satisfactory: the lower bounds that we can give only under the assumption of ETH are much weaker. Rather, our results uncover a surprising link between the parameterized algorithms for \minimumfillin and related problems, and the approximability of \minb. Thus it seems that the question about the optimality of the former has a much deeper, fundamental nature.

Therefore, we believe that our work strongly motivates further investigation of Hypothesis~\ref{h.bisection}. Can this conjecture be linked to ETH and possibly some other strong conjectures like SETH, the existence of linear PCPs, or the conjectures proposed by Feige~\cite{Feige02}? Or maybe it can be simply disproved?

Our improved lower bound for \fast still has a gap between $k^{1/4}$ and $k^{1/2}$. A closer inspection of the proof uncovers a fundamental obstacle for why we cannot achieve tightness: the Chernoff concentration bound used in the proof of Theorem~\ref{thm:only-eth-fast} is essentially tight, because in every tournament on $n$ vertices there is a feedback arc set of size $\binom{n}{2}/2-\Omega(n^{3/2})$~\cite{Hassin94approximationsfor}. If this error term was of magnitude $\Theta(n)$ instead of $\Theta(n^{3/2})$, then our approach would give a tight result for \fastshort. Can this problem be circumvented, or maybe the high anticoncentration of the number of feedback edges in a random ordering of a tournament can be exploited algorithmically to obtain a faster algorithm?

Finally, even assuming Hypothesis~\ref{h.bisection} we do not get tight bounds, due to the (poly)logarithmic factors in the exponent describing the running time of the existing algorithms for completion problems. Bridging this gap can be seen as another forthcoming goal.

\section*{Acknowledgements}
We would like to thank Per Austrin for valuable discussions about different versions of the PCP theorem and  anonymous reviewers for their helpful comments.

%% file: sec_problems.tex
\section{Problem definitions}
\label{sec:problems}

\problem{\ola\ (\olashort)}{
  A graph $G=(V,E)$, an integer $k$.
}{
  Does there exist a linear arrangement $\aorder$ of $G$ of cost at most $k$?
}

\problem{\olad[d] (\oladshort[d])}{
  A graph $G=(V,E)$ with degree at most $d$, an integer $k$.
}{
  Does there exist a linear arrangement $\aorder$ of $G$ of cost at most $k$?
}

\problem{\maxcut}{
  A graph $G=(V,E)$, an integer $k$
}{
  Does there exist a cut of size at least $k$?
}

\gapproblem{$\gapmaxcut_{[\bounda,\boundb]}$}{
  A graph $G=(V,E)$.
}{
  $G$ admits a cut of size at least $\boundb m$.
}{
  $G$ does not admit a cut of size larger than $\bounda m$.
}

\problem{\minb}{
  A graph $G=(V,E)$ with even number of vertices, an integer $k$.
}{
  Does there exist a cut $(A,B)$ of size at most $k$, such that $|A| = |B|$?
}

\gapproblem{\gapminbd[d]$_{[\bounda, \boundb]}$}{
  A $d$-regular graph $G=(V,E)$ with even number of vertices.
}{
  $G$ admits a cut $(A,B)$ of size at least $\boundb m$, such that $|A| = |B|$.
}{
  $G$ does not admit a cut $(A,B)$ of size larger than $\bounda m$, such that $|A| = |B|$.
}

\subsection{Satisfiability problems}

We consider several variants of the satisfiability problem, in general defined as follows.

\problem{\xsat}{
  An \xsat\ formula $\phi = C_1 \wedge \ldots \wedge C_m$.
}{
  Does there exist an assignment of the variables of $\phi$, such that $\phi$ is satisfiable?
}

\gapproblem{$\xsat_{[\bounda,\boundb]}$}{
  An \xsat\ formula $\phi = C_1 \wedge \ldots, \wedge C_m$.
}{
  $\phi$ admits an assignment satisfying at least $\boundb m$ clauses.
} {
  $\phi$ does not admit an assignment assignment satisfying more than $\bounda m$ clauses.
}

Where an \xsat\ formula is a formula from the \justsat\ related problem, precisely it is an E$l$-CNF formula for \esat[$l$], \enaesat[$l$], and an $l$-AND formula for \andsat[$l$].

We also similarly define a problem \xsat$(d)$ with the difference that the variables of an input formula occur in at most $d$ clauses, e.g. \andsat[$l$](d).

\subsection{Completion problems}

The following problems is a generic version of a completion problem to a given graph class $X$.

\problem{{\sc X-Completion}}{
  An undirected graph $G$, an integer $k$.
}{
  Is it possible to add at most $k$ edges to $G$, so that the obtained
    graph belongs to the graph class {\sc X}?
}

We consider the following list of completion problems:
\minimumfillin, \chaincompletion, \properintervalcompletion, \intervalcompletion, \thresholdcompletion, \triviallyperfectcompletion, where \minimumfillin is completion to chordal graphs, and other problems have self-descriptive names.

